\documentclass[12pt,leqno]
{amsart}
\usepackage{amsmath,epsfig,graphicx,color, float}
\usepackage{subcaption}
\usepackage{relsize}
\usepackage{comment}
\usepackage{xcolor}
\usepackage[colorlinks, linkcolor=violet,citecolor=violet]{hyperref}
\usepackage{xurl}
\usepackage{natbib}
\usepackage{algorithm}
\usepackage{algorithmicx}
\usepackage{algpseudocode}

\usepackage[margin=1.2in]{geometry}

\numberwithin{table}{section}
\numberwithin{figure}{section}
\numberwithin{equation}{section}

\definecolor{darkblue}{rgb}{.2, 0.2,.8}
\definecolor{darkgreen}{rgb}{0,0.5,0.3}
\definecolor{darkred}{rgb}{.8, .1,.1}
\newcommand{\red}{\color{black}}

\newcommand{\bfX}{\vect{X}}

\newcommand{\bfpi}{\vect{\pi}}
\newcommand{\bfT}{\mat{T}}
\newcommand{\bft}{\vect{t}}
\newcommand{\bfe}{\vect{e}}

\newcommand{\0}{\mat{0}}

\newcommand{\E}{\mathbb{E}}
\renewcommand{\P }{{\mathbb P}}

\newcommand{\ov}{\overline}

\newcommand{\eqd}{\stackrel{d}{=}}

\newtheorem{lemma}{Lemma}[section]
\newtheorem{theorem}[lemma]{Theorem}
\newtheorem{proposition}[lemma]{Proposition}

\newtheorem{example}[lemma]{Example}

\newtheorem{remark}{Remark}[section]

\usepackage{bm}
\newcommand{\vect}[1]{\pmb{#1}}
\newcommand{\mat}[1]{\boldsymbol{\bm #1}}

\DeclareMathOperator*{\argmax}{arg\,max}

\usepackage{marginnote}

\allowdisplaybreaks

\begin{document}
\bibliographystyle{apalike}
\title{Mortality modeling and regression with matrix distributions}

\author[H. Albrecher]{Hansj\"org Albrecher}
\address{Department of Actuarial Science, Faculty of Business and Economics, University of Lausanne, UNIL-Dorigny,
CH-1015 Lausanne}
\email{hansjoerg.albrecher@unil.ch}

\author[M. Bladt]{Martin Bladt}
\address{Department of Actuarial Science, Faculty of Business and Economics, University of Lausanne, UNIL-Dorigny,
CH-1015 Lausanne}
\email{martin.bladt@unil.ch}

\author[M. Bladt]{Mogens Bladt}
\address{Department  of Mathematics,
University of Copenhagen,
Universitetsparken 5,
DK-2100 Copenhagen,
Denmark}
\email{bladt@math.ku.dk}

\author[J. Yslas]{Jorge Yslas}
\address{Institute of Mathematical Statistics and Actuarial Science,
University of Bern,
Alpeneggstrasse 22,
CH-3012 Bern,
Switzerland}
\email{jorge.yslas@stat.unibe.ch}

\begin{abstract}
In this paper we investigate the flexibility of matrix distributions for the modeling of mortality. Starting from a simple Gompertz law, we show how the introduction of matrix-valued parameters via inhomogeneous phase-type distributions can lead to reasonably accurate and relatively parsimonious models for mortality curves across the entire lifespan. A particular feature of the proposed model framework is that it allows for a more direct interpretation of the implied underlying aging process than some previous approaches. Subsequently, towards applications of the approach for multi-population mortality modeling, we introduce regression via the concept of proportional intensities, which are more flexible than proportional hazard models, and we show that the two classes are asymptotically equivalent. 
We illustrate how the model parameters can be estimated from data by providing an adapted EM algorithm for which the likelihood increases at each iteration. The practical feasibility and competitiveness of the proposed approach, { including the right-censored case, are illustrated by several sets of mortality and survival data.}

\end{abstract}
\maketitle

\section{Introduction}\label{sec:int}
The statistical modeling of human lifetimes is a central topic in actuarial science, as it has immediate consequences for the pricing and management of a wide range of insurance and pension products (see, e.g., \cite{olivieri2015}, \cite{dickson2019}). Starting all the way back with \cite{gompertz1825}, who proposed a mortality rate that increases exponentially over the lifetime, the field has seen a tremendous development over the last 200 years, and nowadays, sophisticated models are available that take into account the changing nature of mortality over time as well as common and individual trends across different (sub)populations (see, e.g., \cite{pitacco2004,pitacco2019} for surveys and \cite{denuit2016}, \cite{dowd2020}, \cite{lin2021correlated}, \cite{shapovalov2021}, \cite{zeddouk2020mean} and \cite{renshaw2021} for some very recent developments in different research directions, and \cite{barigou2021} for a proposition how to combine various alternative models into one). A detailed understanding of these trends is crucial for prudent management of longevity products (\cite{barrieu2012}, \cite{sherris2014})
 and the future of old-age provision on the society level in general (see, e.g., \cite{albrecher2016}). At the same time, 
 the statistical estimation of mortality models from given (often incomplete and interval-censored) mortality data can be quite challenging; see, e.g., \cite{macdonald2018modelling} for a recent account. In the presence of the many available model options for a given purpose (e.g., constructing lifetables or making mortality projections), a number of countries nowadays rely on an accorded effort of academia and insurance practice for respective  recommendations or guidelines (see, e.g., \cite{antonio2017} for an excellent survey on the current approach of Belgium and the Netherlands). 

Most of the models in practical use are based on a purely statistical fitting of observed mortality patterns and subsequently extrapolating the trends of the past into the future, which historically has sometimes led to an underestimation of the actual mortality improvement. In addition, the resulting models typically rely on an enormous number of parameters, like for instance, 300 in the very popular Lee-Carter model (\cite{lee1992modeling}), which employs separate parameters for each year and each age, that then need to be estimated from the given data, cf.\ \cite{li2009}. This may not be considered an issue, as long as the values of these parameters remain reasonably stable during updates, and if for instance cohort effects are thought of dominating the stochastic nature of mortality. At the same time, it can be interesting to see whether, through other and possibly parsimonious approaches, one can also capture observed mortality patterns to a satisfactory degree. 
Motivated by corresponding approaches in medicine and early contributions by, e.g.,  \cite{gutterman1998} in the actuarial context, \cite{lin2007markov} proposed an interesting alternative approach for the modeling of human lifetimes that links its parameters to the biological and physiological mechanisms of aging. Concretely, one may imagine a lifetime as traversing through a number of stages that can be thought of as markers of the biological age, with the time that it takes an individual to go from one stage to the next being random, and an additional possibility to die 'prematurely' at each stage along the way (e.g., due to accidents or other external causes). If such a model can be calibrated to mortality data in a satisfactory way, then one could possibly attribute particular markers to each stage, linking the parameters of the model more directly to the physiological aging mechanism, and  opening up the possibility to include expert opinions on particular mortality trends into the models and the resulting forecasts more directly. In addition, such an approach could also allow to compare the mortality patterns of different (sub)populations on the basis of changes of transition rates between those stages, which is a somewhat appealing alternative to the purely statistical approach that is typically followed today, see \cite{lin2007markov} and also \cite{cheng2020} for a detailed discussion as well as \cite{hassan2014} for an application in a disability context. 

A natural candidate for a model of the aging process of the human body is a finite-state homogeneous continuous-time Markov process with a single absorbing state (death). In the context of modeling stages of a disease, such a Markov approach can, e.g., be traced back to \cite{kay1986}, \cite{longini1989}, \cite{guihenneuc2000}. For early contributions towards human lifetime modeling along these lines, see \cite{gavrilov1991}, and \cite{aalen1995phase} who particularly focussed on the modeling of hazard rates. The resulting class of distributions for the lifetime are known as \textit{phase-type distributions}. Phase-type (PH) distributions are particularly tractable and have been systematically studied as matrix distributions over the years (see \cite{neuts75,neuts1981matrix} and \cite{Bladt2017} for a survey). 
One of the most attractive features of PH distributions from a statistical perspective is that they are dense in the class of distributions on the positive real line in the sense of weak convergence, so that any given distribution on the positive half-line can be approximated arbitrarily well. However, for a good fit to a given histogram of data, often a large number of states (phases) are required, making the estimation procedure complex. Several estimation methods have been proposed in the literature, including the method of moments in \cite{bobbio2005matching,horvath2007matching} and references therein, Bayesian inference in \cite{bladt2003estimation}, and maximum likelihood estimation via the expectation-maximization (EM) algorithm in \cite{asmussen1996fitting} (for the right-censored case see \cite{olsson1996estimation}). The EM method considers the full likelihood arising from the path representation of a PH distribution and has been the most popular approach for many applications in applied probability. 

In \cite{lin2007markov}, the human aging process is modeled by such a PH distribution (concretely, a Coxian distribution). It turns out that for an adequate fit to data (the eventual lifetimes), about 200-250 phases are needed (see also \cite{asmussen2019}, where about 50 phases are suggested in the context of pricing equity-linked insurance products, when the focus is more on the stability of eventual prices rather than the accuracy of the fit). \cite{cheng2020} introduce additional restrictions of the parameters for the fitting procedure. The main reason for the need of so many phases (here and in other applications) towards an adequate fit is typically due to the fact that all PH distributions have exponentially decaying tails, whereas the actual data do not have that property. Rather than using many phases to correct for that, \cite{albrecher2019inhomogeneous} proposed a framework of inhomogeneous Markov processes for the PH construction, where time evolves according to a deterministic transform of the original time, which can be targeted towards capturing the tail behavior of the resulting distribution via the transform rather than the introduction of additional states. The resulting class of  \textit{inhomogeneous phase-type distributions} (IPH) is still dense in the class of all distributions on the positive half-line and turns out to lead to adequate fits with substantially fewer phases (and consequently parameters) in various applications. The statistical estimation of IPH distributions was then subsequently developed in \cite{albrecher2020iphfit}. A number of IPH distributions can be seen as matrix extensions of other base distributions, where the latter determine the tail behavior, and the matrix-valued parameters improve the fit of the formerly scalar counterparts, cf.\ \cite{albrecher2019inhomogeneous}. 

The purpose of this paper is to study the potential of IPH distributions for the modeling of human mortality. In particular, we will look into matrix versions of the Gompertz distribution, which will give the needed additional flexibility for fitting observed mortality rates (hazard functions), yet keeping the number of phases (and consequently parameters) of the model reasonably low. Importantly, the resulting model will allow for an interpretation in terms of the above aging process through phases, with time being transformed according to a deterministic mechanism. The focus in the modeling will be on an adequate representation of the hazard rate. In that way, we combine the original approach of hazard rate modeling with a generalized version of \cite{lin2007markov}. The advantage of this combined approach is that the number of necessary phases for a {\color{black}visually acceptable} fit goes down from 200-250 to 10-12, and the approach does not have to be restricted to the sub-class of Coxian PH distributions. 

In a second step, we will then investigate (for the first time) regression of IPH distributions, which will lead a way to multi-population mortality modeling within the above framework of modeling the aging mechanism. In the life sciences, particular instances of PH distributions (namely Coxian distributions) were considered as survival regression models, within the accelerated failure time (AFT) specification, cf.\ \cite{mcgrory2009fully,rizk2019alternative,tang2012modeling} and references therein. The popularity of these models for several decades is the product of their natural interpretation in terms of states and their statistical efficacy. Markov Chain Monte Carlo (MCMC) methods seem to be particularly popular for this type of application. In this paper, we propose survival regression models based on general IPH distributions. Concretely, we propose a proportional intensities model for the regression, which contains a number of classical methods as special cases (such as the proportional hazards model for any parametric family, the AFT model and Coxian regression). The combination of inhomogeneous time transformation together with general sub-intensity matrices makes our model significantly more flexible both in terms of tails and hazard function shapes. Regressing on the inhomogeneity parameters also leads to matrix generalizations of models such as the one of \cite{hsieh2000sample}. {\color{black}A particular case of our specification was studied in \cite{bladt2021phase} for severity modeling, where right-censoring, aggregation, and regressing on the inhomogeneity parameters were not considered.}

To make the models practically useful, we adapt techniques developed in  \cite{asmussen1996fitting,olsson1996estimation} and \cite{albrecher2020iphfit} to obtain effective estimation procedures based on the EM algorithm. While the focus in this paper is on mortality modeling, the generality of the approach could, with corresponding adaptations of the techniques, also be of interest for other application areas. We would also like to stress that with its great flexibility comes a drawback, which is that our model does not and is not expected to outperform any individual survival regression model specification in the literature. In particular, the message that we aim to convey is that inhomogeneous Markov aging models are reaching a mature state where they actually are able to resemble many of the classical models when it comes to explaining real data. Prediction and other more delicate aspects of mortality modeling (in particular in connection with time effects and cohort effects) are supported but do not outperform the state-of-the-art {  machine learning methods. Nonetheless, our models enjoy favorable estimation and closed-form mathematical formulas for the resulting lifetime distribution. The latter is important to build realistic insurance products that can be analyzed explicitly, and thus the contribution is also targeted to researchers who want to use our model as a building block in their own research (with the advantage that they may calibrate them on the entire lifetime distribution). }


The structure of the remainder of the paper is as follows. In Section \ref{sec:mort}, we revisit the necessary existing background on IPH distributions with an emphasis on the intensity function, which plays a central role in the paper, and establish a number of explicit expressions in the context of mortality modeling with matrix distributions, addressing the main PH classes under consideration in the paper in more detail. Section \ref{sec:prop_int} is devoted to regression and proposes a proportional intensities model. Main properties are derived and it is shown that it is asymptotically equivalent to the proportional hazards model, but more flexible for our purposes. Section \ref{sec:fit-to-IPH} then develops the estimation methodology for the regression models via the EM algorithm. 
The presented methods and models are illustrated in Section \ref{sec:examples} for mortality data from Denmark, Japan, and the USA, both for the univariate case and the regression. In particular, it is shown that for these data dimensions, around 10-12 in the matrix framework are already reasonably competitive when for instance, compared to the classical Lee-Carter model. { For survival data, Section \ref{sec:cens} implements the right-censored version of the main algorithm.} Finally, Section \ref{sec:conclusion} concludes.


%
%

\section{Mortality rate modeling using matrix distributions}\label{sec:mort}
In this section, we recall some basic properties of inhomogeneous phase-type distributions that will be relevant later on. 
%
In the sequel, for a random {\color{black}variable} $X$, we write $X\sim F$ with $F$ a distribution function, density, or acronym, when $X$ follows the distribution uniquely associated with $F$. For two real-valued functions, $g,h$ the terminology $g(t)\sim h(t)$, as $t\to \infty$ is defined as $\lim_{t\to \infty}g(t)/h(t)=1$. Unless stated otherwise, equalities between random objects hold almost surely. 

{\color{black}We will denote by $\mu_f(s)$ the hazard (or mortality) rate of a lifetime random variable $X$ with density $f$}, so that
\begin{align}\label{eq:cdff}
	\overline{F}(x) = \exp \left( -\int_0^x \mu_f(s)ds  \right) 
\end{align} 
is the corresponding survival function, and $\mu_f (s)ds$ can be interpreted as the probability of dying in the time interval $[s,s+ds)$ given that the individual has survived up until time $s$. The variation of $\mu_f (s)$ with $s$ is linked to the aging process. 

{\color{black}In this paper, we make the assumption that
\[ \mu_f (x)=e^{\beta x} \,, \   \]  
for some $\beta>0$, that is, we have an underlying Gompertz hazard rate. For presentation purposes, and also because many of the formulas also work for other hazard rates, we will still often work abstractly with the symbol $\mu_f (x)$ in the sequel.}

{\color{black}Now let} $ ( J_t )_{t \geq 0}$ denote a time-inhomogeneous Mar\-kov jump process on a state-space $\{1, \dots,$ $p, p+1\}$, where states $1,\dots,p$ are transient and state $p+1$ is absorbing. In other words, $ ( J_t )_{t \geq 0}$ has an intensity matrix of the form
\begin{align*}
	\mat{\Lambda}(t)= \left( \begin{array}{cc}
		\bfT(t) &  \bft(t) \\
		\0 & 0
	\end{array} \right)\in\mathbb{R}^{(p+1)\times(p+1)}\,, \quad t\geq0\,,
\end{align*}
where $\bfT(t) $ is a $p \times p$ matrix function and $\bft(t)$ is a $p$-dimensional column vector function. Here, for any time $t\geq0$, $\bft (t)=- \bfT(t) \, \bfe$, where $\bfe $ is the $p$-dimensional column vector of ones. Let $ \pi_{k} = \P(J_0 = k)$, $k = 1,\dots, p$, $\bfpi = (\pi_1 ,\dots,\pi_p )$ be a probability vector corresponding to the initial distribution of the jump process. In particular, we have $\P(J_0 = p + 1) = 0$. Then we say that the time until absorption 
\begin{align*}
	\tau = \inf \{ t \geq  0 \mid J_t = p+1 \}\,,
\end{align*}
has an inhomogeneous phase-type distribution with representation $(\bfpi,\bfT(t) )$. 
For the purpose of this paper, the following particular case is of interest: $\bfT(t) = \lambda(t)\,\bfT$, where $\lambda(t)$ is some known non-negative real function, named the \emph{intensity function}, and {\color{black}$\bfT=\{t_{kl}\}_{1\le k,l\le p}$} is a sub-intensity matrix. {\color{black}Similarly, we let $-\mat{T}\bfe =\bft=\{t_k\}_{1\le k\le p}$.} {\color{black}In practice, we will usually obtain $\lambda$ as a hazard rate from a known distribution, for instance $\lambda=\mu_f$.} We  write $$\tau \sim  \mbox{IPH}(\bfpi , \bfT , \lambda )\,.$$

Note that for  $\lambda(t) \equiv 1$ one returns to the time-homogeneous case, which corresponds to the conventional phase-type distributions with notation $\mbox{PH}(\bfpi , \bfT )$; a comprehensive account of PH distributions can be found in \cite{Bladt2017}.
It is not difficult to show that if $Y \sim  \mbox{IPH}(\bfpi , \bfT , \lambda )$, then there exists a function $g$ such that \begin{equation}\label{gtrans}
Y \sim g(Z) \,,
\end{equation}where $Z \sim \mbox{PH}(\bfpi , \bfT )$.
 Specifically, $g$ is defined in terms of $\lambda$ by
\begin{equation*}
g^{-1}(y) = \int_0^y \lambda (s)ds \,,  \quad y\ge 0\, , \label{eq:transformation-g}
\end{equation*}
or, equivalently,
\begin{equation*}
\lambda (y) = \frac{d}{dy}g^{-1}(y) \,.  \label{eq:transformation-lambda} 
\end{equation*}
To avoid degeneracies, we assume that $g^{-1}(y)<\infty,\:\forall y\ge0$, and $\lim_{y\to\infty}g^{-1}(y)=\infty$.
The density $f_Y$ and survival function $\overline{F}_Y$ of $Y \sim  \mbox{IPH}(\bfpi , \bfT , \lambda )$ are given by
\begin{eqnarray}
 f_Y(y) &=& \lambda (y)\, \vect{\pi}\exp \left( \int_0^y \lambda (s)ds\ \mat{T} \right)\vect{t}\,,\quad y\ge0 \,, \label{eq:dens-IPH}\nonumber \\
 \overline{F}_Y(y)&=& \vect{\pi}\exp \left( \int_0^y \lambda (s)ds\ \mat{T} \right)\vect{e}\,,\quad y\ge0 \,. \label{sy} \label{eq:cdf-IPH}
\end{eqnarray}
For further properties on IPH distributions and motivation for their use in statistical modeling, we refer to \cite{albrecher2019inhomogeneous}. See also \cite{albrecher2020iphfit} for extensions to the multivariate case.

The tail behavior of IPH distributions can be derived from the corresponding tail behavior of a PH distribution. Recall that the survival function of a PH distribution is given by
\begin{align}\label{PHtail_expansion}
\overline{F}_Z(y)=\sum_{j=1}^{m} \sum_{l=0}^{\kappa_{j}-1}y^{l} \exp({\Re\left(-\eta_{j}\right) y })\left[a_{j l} \sin \left(\Im\left(-\eta_{j}\right) y \right)+b_{j l} \cos \left(\Im\left(-\eta_{j}\right) y \right)\right] \,,\quad y\ge0\,,
\end{align}
where $-\eta_j<0$ are the eigenvalues of the Jordan blocks $\mat{J}_{j}$ of $\bfT$, with corresponding dimensions $\kappa_j$, $j = 1,\dots, m$, and $a_{jl}$ and $b_{jl}$ are constants depending on $\bfpi$ and $\bfT$. If $-\eta$ is the largest real eigenvalue of $\bfT$ and $n$ is the dimension of the Jordan block of $\eta$, then it is easy to see from \eqref{PHtail_expansion} that
\begin{align}\label{PHtail_asymptotic}
	\overline{F}_Z(y)\sim c y^{n -1} \exp({-\eta y}) \,, \quad y \to \infty \,,
\end{align}
where $c$ is a positive constant. That is, all PH distributions have exponential tails with Erlang-like second-order bias terms. The tail behavior of $Y$ then follows from \eqref{gtrans}.

A number of IPH distributions can be expressed as classical distributions with matrix-valued parameters. For the representation of such distributions, we make use of functional calculus. If $\red\nu$ is an analytic function and $\mat{A}$ is a matrix, define 
\begin{align*}
	\red \nu( \mat{A})=\dfrac{1}{2 \pi i} \oint_{\Gamma}\nu(w) (w \mat{I} -\mat{A} )^{-1}dw \,,
\end{align*}
where $\Gamma$ is a simple path enclosing the eigenvalues of $\mat{A}$; cf.\ \cite[Sec.3.4]{Bladt2017} for details. In particular, the matrix exponential can be defined in this way.

{\color{black}The matrix distribution of particular interest for the mortality modeling in this paper is defined through the following survival function:}

{\color{black}\[   \overline{H}(x) = \vect{\pi} \exp \left( \mat{T} \int_0^x \mu_f(s)ds  \right)\vect{e}=\vect{\pi}e^{\mat{T}\left(e^{\beta x}-1  \right)/\beta} \vect{e} \,.  \]}
Here, $(\vect{\pi},\mat{T})$ is a $p$-dimensional PH representation (so that the underlying Markov process has $p$ transient states).
In view of \eqref{sy}, this is the survival function of the time until absorption of a time-inhomogeneous Markov process initiating according to $\vect{\pi}$ and with intensity matrix 
\begin{equation}
  \mat{\Lambda}(t) = \mu_f(t) \begin{pmatrix}
\mat{T} & \vect{t} \\
\vect{0} & 0
\end{pmatrix} \,. \label{eq:int-mat}
\end{equation}
Hence $\red H\sim\mbox{IPH}(\vect{\pi},\mat{T},\mu_f)$.
If $\mat{T}$ is one-dimensional taking the value $-1$, then $\red H$ coincides with $F$ {\color{black} in \eqref{eq:cdff}.}

This \emph{matrix-Gompertz} distribution  $\red H\sim \operatorname{mGompertz}(\beta,\vect{\pi},\mat{T})$ has tails much lighter than exponential. Observe that from an inhomogeneous Markov point of view, the sub-intensity matrix of the underlying jump process is given by $\mat{T}\exp({\beta x})$, which for scalar $\mat{T}$ leads back to the exponential hazard function of the Gompertz distribution, the latter being the classical model for human mortality over long age ranges
(see, e.g., \cite{macdonald2018modelling}). 

The density function $\red h(x)$ corresponding to $\red H$ is given by
{\color{black}\begin{align}\label{density_of_model}
h(x) = \vect{\pi} \exp \left( \mat{T} \int_0^x \mu_f(s)ds  \right)\vect{t} \mu_f(x)=e^{\beta x} \vect{\pi} e^{\mat{T}\left(e^{\beta x}-1\right) / \beta} \vect{t}\,,
\end{align}}
with corresponding hazard function 
\begin{equation}\label{rates}  \red\mu_h(x) = \mu_f(x) \frac{ \vect{\pi} \exp \left( \mat{T} \displaystyle \int_0^x \mu_f(s)ds  \right)\vect{t} }{ \vect{\pi} \exp \left( \mat{T} \displaystyle\int_0^x \mu_f(s)ds  \right)\vect{e}} \,. \end{equation}

\begin{remark}\normalfont
One way to justify this choice of model is to argue that  $\red \overline{H}(x)$ allows to improve a base model $\overline{F}(x)$ in a parsimonious way. Indeed, since IPH distributions are dense in the class of distributions on the positive real line (in the sense of weak convergence, cf.\ \cite{albrecher2019inhomogeneous}), one can obtain an arbitrarily close approximation to a given lifetime distribution by adding sufficiently many phases to the state space, and in contrast to the homogeneous case, this can typically be achieved with very few states when the base model is already a good first approximation. In addition, this modeling approach also gives rise to the additional interpretation that life (aging) goes through various  
phases, prone to an overall time-inhomogeneous variation represented by $\mu_f(t)$.
\end{remark}

Let {\red $(J_t)_{t\geq0}$ be} the underlying Markov jump process with intensity matrix \eqref{eq:int-mat}. Its transition probabilities are given by
\[ \color{black}  p_{kl}(s,t) =  \P (J_t=l|J_s=k) = \left[ e^{\mat{T}\int_s^t \mu_f(u)du} \right]_{kl} , \]
for $\red k,l=1,\dots,p$ being transient (non-absorbing) states.  
Now let $\red V_l$ denote the time spent in state $\red l$ during a person's life prior to death and let
\[  \red u_{kl} = \E_k (V_l) = \E (V_l | J_0=k)\,, \quad k,l=1,\dots,p \,, \]
with $\red \mat{U}=\{ u_{kl}\}_{k,l=1,\dots,p}$ denoting the corresponding matrix, which is referred to as the Green matrix. Then with $\red \tau \sim H$, 
{\red 
\begin{eqnarray*}
  u_{kl} &=& \E_k \left( \int_0^\tau 1\{  J_s=l \} ds \right)  \\
  &=&\int_0^\infty \P_k (J_s=l,\tau\geq s)ds \\
  &=&\int_0^\infty  \left[ e^{\mat{T}\int_0^s \mu_f(u)du} \right]_{kl}ds \,,
\end{eqnarray*}
and
\begin{equation}
  \mat{U} = \int_0^\infty e^{\mat{T}\int_0^s \mu_f(u)du} ds \,. \label{eq:Green}
  \end{equation}
  }
Again, for the one-dimensional case with $\mat{T}=-1$, we have that
\[  \mat{U} =  \int_0^\infty e^{-\int_0^s \mu_f(u)du} ds = \int_0^\infty \ov{F}(s)ds \, ,  \] 
the mean of $F$.
If $\mu_f(s)=1$, then $\mat{U}=(-\mat{T})^{-1}$. 

{\color{black} 
For a Gompertz hazard rate function $\mu_f(s)$, define a function $g$ in terms of its inverse by
 \begin{equation}
   g^{-1}(s) = \int_0^s \mu_f(u)du =\frac{1}{\beta}\left(  e^{\beta s} -1  \right), \label{eq:h-function} 
   \end{equation}
   or, equivalently, 
\[   g(s) = \frac{\log (\beta s+1)}{\beta} \,.
  \]}
{\red Note that $g^{-1}(s) \rightarrow +\infty$ as $s\rightarrow +\infty$.}
  Then for a $\eta>0$, using both substitution and partial integration, we get
 {\red\begin{eqnarray*} 
  \int_0^\infty e^{-\eta \int_0^s \mu_f(u)du} ds &=& \int_0^\infty e^{-\eta g^{-1}(s)}ds\\
   &=& \int_0^\infty e^{-\eta s}g^\prime (s)ds\\
    &=& \eta\int_0^\infty e^{-\eta s}g(s)ds\\
     &=& {\color{black}\eta L_g(\eta)=\eta\frac{{\mathrm e}^{\frac{\eta}{\beta}} \mathrm{Ei}_{1}\! \left(\frac{\eta}{\beta}\right)}{
\beta \eta}\,, }
\end{eqnarray*}}
where $L_g$ denotes the Laplace transform of $g$, and where 
\[  \mbox{Ei}_1(z) = \int_1^\infty u^{-1}e^{-uz}du  = \int_z^\infty u^{-1}e^{-u}du . \] Correspondingly, \eqref{eq:Green} reduces to
 \begin{equation}
 \red   \mat{U} = \int_0^\infty e^{\mat{T}\int_0^s \mu_f(u)du} ds = -\mat{T}L_g(-\mat{T}) = L_g(-\mat{T})(-\mat{T})\,  \label{eq:Green-result},
   \end{equation} 
   which can then be written as 
   \begin{eqnarray*}
\mat{U}&=& -\mat{T} \beta^{-1}(-\mat{T})^{-1}e^{-\beta^{-1}\mat{T}} \mbox{Ei}_1 \left( -\beta^{-1} \mat{T} \right)\\
&=&\beta^{-1}e^{-\beta^{-1}\mat{T}} \mbox{Ei}_1 \left( -\beta^{-1} \mat{T} \right),
\end{eqnarray*}
   since $-\mat{T}$ and $L_g(-\mat{T})$ commute. 
{\red More generally,} $\mat{U}$ can be expressed in an explicit form whenever $L_g$ can (the matrix function $L_g(-\mat{T})$ may be computed in different ways in practice, cf.\ \cite{Bladt2017}). {\color{black}We then have the following consequence of the above formulas:
\begin{lemma}
For $X\sim H$ we have
\[ \E (X)=\vect{\pi}L_g(-\mat{T})\vect{t}=\beta^{-1}\vect{\pi}e^{-\beta^{-1}\mat{T}}\mbox{Ei}_1\left( -\beta^{-1}\mat{T}\right)\mat{e}\,. \]
\end{lemma}
\begin{proof}
The expected value of the distribution $H$ can be obtained as the weighted sum of conditional expected times in the different states, where the weights correspond to the initial probabilities. Hence
\[   \E (X) = \vect{\pi}\mat{U}\vect{e} = \vect{\pi}L_g(-\mat{T})(-\mat{T}\vect{e}) = \vect{\pi}L_g(-\mat{T})\vect{t} \,,\]
and we may use the other expression for $\mat{U}$ to complete the proof.
\end{proof}}

{\red Note that for $X\sim H$}, residual lifetimes given survival up to time $x$, say, are readily obtained by the survival function
\[   \P (X>x+y | X>x) = \frac{\vect{\pi} \exp \left(\mat{T} \int_0^x \mu_f(u)du  \right)}{\vect{\pi} \exp \left(\mat{T} \int_0^x \mu_f(u)du  \right)\vect{e}} \exp \left(\mat{T} \int_x^{x+y} \mu_f(u)du  \right) \vect{e} \,. \]
Hence the residual lifetime distribution given survival until time $x$ is again IPH with initial vector 
\[   \vect{\alpha}^x = \frac{\vect{\pi} \exp \left(\mat{T} \int_0^x \mu_f(u)du  \right)}{\vect{\pi} \exp \left(\mat{T} \int_0^x \mu_f(u)du  \right)\vect{e}}\,, \]
intensity matrix $\mat{T}$, and hazard function $\mu_f^x(u)=\mu_f(x+u)$.

{\red 
We now establish a relationship between the intensity function of an IPH distribution and its hazard function.
Note that the intensity function $\lambda$ of a $\mbox{IPH}(\bfpi , \bfT , \lambda )$ distribution is not the same as its hazard function, $\mu$. The latter is given by
	\begin{align*}
	\mu(t) = \lambda(t) \frac{\vect{\pi}\exp \left( \int_0^t \lambda (s)ds\ \mat{T} \right)\vect{t} }{\vect{\pi}\exp \left( \int_0^t \lambda (s)ds\ \mat{T} \right)\vect{e} }  \,.
	\end{align*}
This distinction is subtle, but leads to increased flexibility in the interplay of hazards between subgroups, as we will discuss in Section~\ref{sec:prop_int}. The cumulative hazard function of an $\mbox{IPH}(\bfpi , \bfT , \lambda )$ distribution reads 
\begin{align*}
	M(t) & =  \int_{0}^{t} \mu (s) ds =- \log \left(\vect{\pi}\exp \left( \int_0^t \lambda (s)ds\ \mat{T} \right)\vect{e}\right) \,.
\end{align*}

The next result shows that $\lambda$ is asymptotically equivalent to the hazard function in the upper tail. 
\begin{theorem}\label{aysmpt_proportional}
	Let $ Y \sim\mbox{IPH}(\bfpi , \bfT , \lambda )$. Then the hazard function $\mu$ and cumulative hazard function $M$ of $Y$ satisfy, respectively,  
	\begin{align*}
		\mu(t) &\sim c \lambda(t) \,, \quad t \to \infty \,,\\
		M(t)& \sim k g^{-1}(t) \,, \quad t \to \infty \,,
	\end{align*}
	where $g^{-1}(t)=\int_0^t \lambda(s)ds$, and $c,k$ are positive constants.
\end{theorem}
\begin{proof}
	 We have that if $Z\sim \mbox{PH}(\bfpi,\bfT)$, then its density $f_{Z}$ has a representation of the same form as \eqref{PHtail_expansion} (cf. \cite[Sec.4.1.1]{Bladt2017} for details), which only differs on the multiplicative constants. Thus,
\begin{align*}
\overline{F}_Y(y)&\sim c_1 [g^{-1}(y)]^{n -1} \exp\left({-\eta [g^{-1}(y)]}\right)\,,\\
f_Y(y) &\sim c_2 [g^{-1}(y)]^{n -1} \exp\left({-\eta [g^{-1}(y)]}\right)\frac{d}{dy} [g^{-1}(y)]\,,\
\end{align*}  
as $y \to \infty$, where  $-\eta$ is the largest real eigenvalue of $\bfT$, $n$ is the dimension of the Jordan block of $\eta$, and $c_1,c_2$ are positive constants. The first expression then follows by recalling that $\frac{d}{dy} [g^{-1}(y)]=\lambda(y)$. The second one can be shown in an analogous manner.
\end{proof}
}

\subsection{Special {\color{black} Coxian} sub-structures}
The analysis becomes particularly tractable when the underlying PH distribution $(\vect{\pi},\mat{T})$ has a simple form. {\red For example,}
%
 a  Coxian PH distribution can be written in the form
 \[  \vect{\pi}=(1,0,\dots,0), \ \ \mat{T}=
 \begin{pmatrix}
 -\lambda_1 & \nu_1\lambda_1 & 0 & \dots & 0 \\
 0 & -\lambda_2 & \nu_2\lambda_2  & \dots & 0 \\
 0 & 0 & -\lambda_3 & \dots & 0 \\
 \vdots & \vdots & \vdots & \vdots\vdots\vdots & \vdots\\
 0 & 0 & 0 & \dots & -\lambda_p
 \end{pmatrix} ,
    \]
    where all $\lambda_k$, $k=1,\dots,p,$ are distinct. The phase diagram for such a distribution is given by 
    \begin{center}
    	\setlength{\unitlength}{0.27mm}
    	\begin{picture}(400,90)
    		\put(-20,35){\framebox(30,30){$1$}}
    		\put(-5,35){\vector(0,-1){30}}
    		\put(435,35){\framebox(45,30){$p+1$}}
    		\put(-60,50){\vector(1,0){40}}  		
    		\put(10,50){\vector(1,0){60}}
    		\put(25,60){\tiny $\nu_1\lambda_1$}
    		\put(100,50){\vector(1,0){60}}
    		\put(115,60){\tiny$\nu_2\lambda_2$}
    		\put(70,35){\framebox(30,30){$2$}}
    		\put(0,20){\tiny$(1-\nu_1)\lambda_1$}
			\put(90,20){\tiny$(1-\nu_2)\lambda_2$}
			\put(85,35){\vector(0,-1){30}}
				\put(410,60){\tiny$\lambda_p$}
					\put(190,60){\tiny$\nu_{p-2}\lambda_{p-2}$}
				\put(250,35){\framebox(40,30){$p-1$}}
			\put(395,50){\vector(1,0){40}} 
			\put(290,50){\vector(1,0){75}}  
			\put(190,50){\vector(1,0){60}}  
			\put(165,35){$\cdots$} 		
			\put(300,60){\tiny$\nu_{p-1}\lambda_{p-1}$}
			\put(365,35){\framebox(30,30){$p$}}
			\put(275,20){\tiny$(1-\nu_{p-1})\lambda_{p-1}$}
				\put(270,35){\vector(0,-1){30}}
    	\end{picture}\end{center}
That is, the states are traversed sequentially, and at each state, there is a probability to exit to the absorption state $p+1$ (i.e., die) directly. In this case
    \[  (z\mat{I}-\mat{T})^{-1} =
 \begin{pmatrix}
 \frac{1}{z+\lambda_1} & \frac{\nu_1\lambda_1}{(z+\lambda_1)(z+\lambda_2)} & \frac{\nu_1\lambda_1\nu_2\lambda_2}{(z+\lambda_1)(z+\lambda_2)(z+\lambda_3)} & \cdots & \frac{\nu_1\lambda_1\cdots \nu_{p-1}\lambda_{p-1}}{(z+\lambda_1)(z+\lambda_2)\cdots (z+\lambda_p)} \\
 0 & \frac{1}{z+\lambda_2} & \frac{\nu_2\lambda_2}{(z+\lambda_2)(z+\lambda_3)} & \cdots & \frac{\nu_2\lambda_2\cdots \nu_{p-1}\lambda_{p-1}}{(z+\lambda_2)(z+\lambda_3)\cdots (z+\lambda_p)} \\
 0 & 0 & \frac{1}{z+\lambda_3} & \cdots & \frac{\nu_3\lambda_3\cdots\nu_{p-1} \lambda_{p-1}}{(z+\lambda_3)(z+\lambda_4)\cdots (z+\lambda_p)} \\
 \vdots & \vdots & \vdots & \vdots\vdots\vdots & \vdots\\
 0 & 0 & 0 & \cdots & \frac{1}{z+\lambda_p}
 \end{pmatrix},
      \]
  so with $\nu_p=0$ we get the density
  \begin{eqnarray*}
 \vect{\pi} \exp(\mat{T}x)\vect{t}&=&\frac{1}{2\pi  i }\int_\gamma \exp(s)\vect{\pi}(s\mat{I}-x \mat{T})^{-1}\vect{t}d s  \\
 &=&  \frac{1}{2\pi  i }\int_\gamma \exp(s) \left[\sum_{k=1}^p
 \frac{(\lambda_1\nu_1x)(\lambda_{2}\nu_{2}x)\cdots (\lambda_{k-1}\nu_{k-1}x)\lambda_k(1-\nu_k)}{(s+\lambda_1x)(s+\lambda_{2}x)\cdots (s+\lambda_kx)}
  \right]d s \\
 &=& \sum_{k=1}^p\left(\lambda_k(1-\nu_k)\prod_{m=1}^{k-1}\lambda_m\nu_m \right)x^{k-1}
 \frac{1}{2\pi  i }\int_\gamma  \frac{\exp(s)}{(s+x\lambda_1)\cdots (s+x\lambda_k)}d s \\
 &=&\sum_{k=1}^p\left(\lambda_k(1-\nu_k)\prod_{m=1}^{k-1}\lambda_m\nu_m \right)x^{k-j}
 \sum_{m=1}^k \frac{\exp(-\lambda_m x)}{\displaystyle\prod_{\stackrel{n=1}{n\neq m}}^k (-x \lambda_m+x \lambda_n)}\\
 &=&  \sum_{k=1}^p\left(\lambda_k(1-\nu_k)\prod_{m=1}^{k-1}\lambda_m\nu_m \right)
 \sum_{m=1}^k \frac{\exp(-\lambda_m x)}{\displaystyle\prod_{\stackrel{n=1}{n\neq m}}^k (\lambda_n-\lambda_m)} \,.
 \end{eqnarray*}  

\begin{remark}\normalfont
If $\vect{\pi}$ is relaxed to be any initial probability vector (i.e., the Markov jump process can also start in other states than the first one), then one speaks of a {\em generalized Coxian distribution}. While for the modeling of the aging process as discussed in Section \ref{sec:int} the classical Coxian construction is particularly appealing as it has the direct interpretation of traversing through the $p$ stages of aging, but with a positive probability to die 'prematurely' during one of these states, the generalized Coxian construction is also of interest. In particular, if it leads to a much better fit than the classical Cox construction, it may suggest a significant heterogeneity in the population under consideration.
\end{remark}

%

\section{Regression with proportional intensities}\label{sec:prop_int}
In a next step, we now introduce regression in our PH modeling framework. Regression models for life tables were cast into the scene by the introduction of the proportional hazards models, cf.\ \cite{cox1972regression}, and these models have remained popular, both in their parametric and non-parametric forms up to today. In the logarithmic scale, the specification implies, however, parallel log-mortality curves, which are typically not observed in datasets for which a joint mortality modeling is of interest.  We, therefore, look for introducing regression in a less restrictive way, while still allowing for a direct interpretation of the regression coefficients. This is achieved by letting the intensity function of our IPH distribution (rather than the hazards) vary proportionally with the covariates. As we will see later, this leads to a more flexible model, which is asymptotically equivalent to the proportional hazards construction. A physical interpretation of this construction is that each population subgroup has a natural clock running at a proportional speed relative to other subgroups. 
 
We, therefore, propose a regression model for IPH distributions with representation $\mbox{IPH}(\bfpi , \bfT , \lambda )$, where the intensity $ \lambda ( \,\cdot\, ; \vect{\theta} )$ is  a non-negative parametric function depending on the vector $\vect{\theta}$ and incorporate the predictor variables $\bfX = (X_1, \dots, X_d) $
 by specifying
\begin{align}\label{prop_intens_spec}
	\lambda(t \mid  \bfX , \vect{\theta} )=\lambda ( t ;  \vect{\theta} ) m(\bfX \vect{\beta} ) ,\quad t\ge0.
\end{align}
Here $ \vect{\beta}$ is a $d$-dimensional column vector and $m$ is any {\color{black} positive-valued and }measurable function. We call this model the {\em proportional intensities} (PI) model. The interpretation in terms of the underlying inhomogeneous Markov structure is that time is changed proportionally for a given subgroup when time-transforming the associated PH distribution into an IPH one.  In the sequel, we will indistinctly denote by $\mat{X}$ a vector of covariates at the population level, or a matrix of covariates for finite samples. In the latter case, $\mat{X}_i$ denotes the $i$-th row of the matrix, corresponding to the $i$-th individual in the sample.

Note that the density $f$ and survival function $\overline{F}$ of a random variable with distribution $$ \mbox{IPH}(\bfpi , \bfT , \lambda(\,\cdot\, \mid \bfX , \vect{\theta}) )$$ are given by
 \begin{eqnarray*}
 f(y) &=& m(\bfX \vect{\beta} ) \lambda (y; \vect{\theta})\, \vect{\pi}\exp \left( m(\bfX \vect{\beta} )\int_0^y \lambda (s; \vect{\theta})ds\ \mat{T} \right)\vect{t} \,, \label{eq:dens-IPH} \\
 \overline{F}(y)&=& \vect{\pi}\exp \left(m(\bfX \vect{\beta} ) \int_0^y \lambda (s; \vect{\theta})ds\ \mat{T} \right)\vect{e} \,.  \label{eq:cdf-IPH}
\end{eqnarray*}

\begin{remark}\rm
Since the exponential function $m(x)=\exp(x)$ in the above specification is the most common choice (for instance, the one used in the original Cox proportional hazards model, \cite{cox1972regression}), the PI model draws its name from that particular example. Other choices also go under the proportional name, see for instance \cite{royston2002flexible}. {\color{black} In terms of coefficients, the interpretation of $\vect{\beta}$ is akin to the parameters in a regular Cox regression, except that the effects here are on the intensity of a Markov process, which is only asymptotically equivalent (but not equal) to the hazard rate.}
\end{remark}

\begin{remark}\rm
Within the above setup, it is possible to add dependence of $g$ on $\mat{X}$ also. That is, we may consider the model
\begin{align}\label{double_prop_intens_spec}
	\lambda(t \mid  \bfX , \vect{\gamma} )=\lambda ( t ;  \vect{\theta}(\mat{X}\vect{\gamma}) ) m(\bfX \vect{\beta} ),\quad t\ge0,
\end{align}
where $\vect{\theta}(\mat{X}\vect{\gamma})$ is a vector-valued function, mapping the score $\mat{X}\vect{\gamma}$ to the parameter space of $\lambda$. For instance, in the one-dimensional case, $$\vect{\theta}(\mat{X}\vect{\gamma})=\exp(\gamma_0+\mat{X}\vect{\gamma})$$is a natural choice. In the sequel, any $\vect{\theta}$ specification may be of the form $\vect{\theta}(\mat{X}\vect{\gamma})$, and the notation may be used interchangeably, when there is no risk of confusion. {\color{black} In this specification, the role of $\vect{\theta}$ is to regulate the Gompertz time-transformation of the underlying Markov dynamics. The incorporation of covariates thus means that shortening or elongating virtual times is necessary to describe the differences observed in the data for individuals of differing characteristics.}
\end{remark}


%
%

\begin{remark}\rm
	In the proportional hazards model, one assumes that the hazard function satisfies
\begin{align*}
\red 	\mu(t \,|\, \bfX )=\mu(t) \exp(\bfX \vect{\beta} )  \,,
\end{align*}
with $\red \mu(t) $ referred to as the {\em baseline hazard function} or the {\em hazard function for a standard subgroup} (with $\bfX \vect{\beta} =0$). Common parametric models for $\red \mu(t)$ in the context of mortality modeling are the exponential function (referring to the Gompertz assumption) and other more refined approximations employed for the fitting of empirical mortality curves (cf.\ \cite{kostaki1992}, \cite{macdonald2018modelling}). 
\end{remark}

One of the advantages of the PI model is that the implied hazard functions between different subgroups can deviate from proportionality in the body of the distribution (for $\mbox{dim}(\mat{T})>1$), as can be observed in Section \ref{sec:examples}. This addresses one of the common practical drawbacks of the proportional hazards model. 
If $p=1$, we have that $\mu(t) = c \lambda(t)$ for a constant $c>0$, so that the PI specification reduces to the proportional hazards model for this case. 

%

Another important property of the PI model is that a random  variable following this specification can be expressed as a functional of a PH random variable. More specifically, consider $Y \sim \mbox{IPH}(\bfpi , \bfT , \lambda(\,\cdot\, \mid \bfX , \vect{\gamma}) )$ and let $g(\, \cdot \,|\, \bfX, \vect{\gamma})$ be defined in terms of its inverse function  
$$g^{-1}(y \,|\, \bfX, \vect{\gamma})=\int_{0}^{y}\lambda(s \,|\, \bfX, \vect{\gamma})ds = m(\bfX \vect{\beta} )\int_{0}^{y}\lambda(s; \vect{\theta}(\bfX\vect{\gamma}))ds\,,$$ 
 Then it is not hard to see that
\begin{align*}
	Y \eqd g(Z\,|\, \bfX, \vect{\gamma}) \,,
\end{align*}
for any $Z\sim \mbox{PH}(\bfpi, \bfT)$.
Conversely, 
\begin{align}\label{eq:PHrep}
	Z=g^{-1}(Y \,|\, \bfX, \vect{\gamma}) \sim \mbox{PH}(\bfpi , \bfT )\,.
\end{align}
 This distributional property is a key observation for the estimation procedure proposed in Section~\ref{sec:fit-to-IPH}.

Next, we concentrate on two important distributions which, in the presence of covariates, generalize the commonly employed models for parametric proportional hazards modeling: the exponential and the Weibull baseline hazard models. 


\begin{example}[PH regression] \rm 
	Consider $\lambda \equiv 1$, thus $\lambda(t \,|\,  \bfX )= m(\bfX \vect{\beta} ) $ for all $t >0$. The survival function takes the form 
	\begin{align*}
		\overline{F}(y)= \vect{\pi}\exp \left(m(\bfX \vect{\beta} ) \mat{T} y \right)\vect{e} \,. 
	\end{align*}
Note also that $Y \eqd {Z}/{m(\bfX \vect{\beta} )}$, where $Z \sim \mbox{PH}(\bfpi, \bfT)$. In particular, this implies that 
\begin{align*}
	\E (Y) = \frac{\E(Z)}{m(\bfX \vect{\beta} )}  = \frac{\bfpi (- \bfT)^{-1} \bfe}{m(\bfX \vect{\beta} )}  \, .
\end{align*}
\end{example}

\begin{example}[Matrix-Weibull regression] \rm 
	Consider $\lambda(t) = \theta t^{\theta-1}$, thus $\lambda(t \, | \,  \bfX )= m(\bfX \vect{\beta} ) \theta t^{\theta-1} $ for all $t >0$. The survival function takes the form 
	\begin{align*}
		\overline{F}(y)= \vect{\pi}\exp \left(m(\bfX \vect{\beta} ) \mat{T} y^{\theta} \right)\vect{e} \,. 
	\end{align*}
Note also that $Y \eqd ( Z/m(\bfX \vect{\beta} ))^{1/\theta} = m(\bfX \vect{\beta} )^{-1/\theta} Z^{1/\theta}$, where $Z \sim \mbox{PH}(\bfpi, \bfT)$. The expression for the expected value is then given by 
\begin{align*}
	\E (Y) = m(\bfX \vect{\beta} )^{-1/\theta}  \E (  Z^{1/\theta})  = m(\bfX \vect{\beta} )^{-1/\theta} \Gamma(1 + 1/\theta)  \bfpi (- \bfT)^{-1/\theta} \bfe.
\end{align*}
\end{example}

\begin{remark}\rm 
For $p=1$, we recover the conventional proportional hazards models with exponential and Weibull baseline hazards.
	Moreover, both cases have the interpretation that the expected values are proportional among subgroups. That is, for arbitrary matrix dimension, when the baseline hazard is PH or matrix-Weibull, the PI model is equivalent to the \textit{accelerated failure time} (AFT) model (cf. \cite{pike1966method} for the Weibull case, and more generally, the log-location scale model in \cite{lawless2011statistical}). It is also not hard to see that homogeneous functions $g$ 
	simultaneously satisfy both models.
	However, for our purposes the AFT model is not of primary interest, since it regresses proportionally on the means instead of on the underlying intensity, and the physical interpretation  in terms of hidden Markov models would typically be lost.
\end{remark}
\section{Parameter estimation}\label{sec:fit-to-IPH}
To estimate the hazard rate of a distribution, there are several approaches one may take. It is common in mortality modeling to consider the problem as a regression on the mortality or log-mortality curve, as a function of age, and with an appropriate loss function. However, for PH distributions, such an approach can only be used if very good (or near-optimal) initial parameters are supplied. Otherwise, local maxima and boundaries of the parameter space are almost unavoidable, and fitting times may be very long. 

With that in mind, we present in this section two estimation approaches for the PI model. The first one is the target approach, which deals with the mortality curve directly and is a general procedure, here developed without covariates. It can also be seen as an alternative to the parametric mortality curve fitting in the spirit of \cite{heligman1980}, \cite{kostaki1992}, now for the hazard rate of a matrix-Gompertz distribution. The second is a maximum-likelihood approach via the EM algorithm which provides a fast way of obtaining good initial parameters for the first method\footnote{\color{black}However, also note that if the modeler is interested in tail probabilities instead of mortality estimates, it is much preferable to exclusively use the EM algorithm to directly estimate the density function. The alternative, that is, obtaining survival estimates through the integrated hazard, will incur model errors that propagate from younger into older ages.}$^{,}$\footnote{{\color{black}The initial parameters of the EM algorithm itself are of lesser concern since the routine is much faster and more robust. In all our examples, for the EM algorithm, we randomly simulated a valid PH structure (initial vector summing to one, and sub-intensity matrix having negative diagonal and rows summing to less than zero) as initial parameters, and the number of EM iterations was $2000$, which was sufficient for the final changes in likelihood to be less than $0.01\%$ in all cases.}}. The main conceptual difference between the two approaches is that the first one gives equal weight to all ages, while the second gives weights proportional to their abundance in the dataset.

\subsection{Direct mortality modeling} 
We define a global loss function $\ell$ as the aggregated age-specific losses $L$, which is a function of the observed and fitted mortality curves, as follows: 
\begin{align*}
\ell(\vect{\pi},\mat{T},f | \vect{\mu})=\sum_{x=0}^N L\left(\mu_h(x),\mu_x\right),
\end{align*}
where $\mu_x$ is the observed mortality at age $x$ and $\mu_h$ is given by \eqref{rates}. Several choices for $L$ have been proposed, arising either from graphical considerations (for instance, Euclidean or $L_p$ norms, absolute or relative differences) or from probabilistic considerations on the counts of the life table itself (for instance, Poisson, cf. \cite{brostrom1985practical}, or other count distributions such as Binomial or Negative Binomial). The empirical work for the present manuscript yielded that the sturdiest choice for the PI model is given by $L(\mu,\nu)=(\log(\mu)-\log(\nu))^2$. In that case, we may apply the following decomposition:
\begin{align*}
\ell(\vect{\pi},\mat{T},f | \vect{\mu})=\sum_{x=0}^N (\log(\mathcal{C}(x|\vect{\pi},\mat{T},f))+\log(\mu_f(x))-\log(\mu_x))^2,
\end{align*}
with the correction factor
\begin{align*}
\mathcal{C}(x|\vect{\pi},\mat{T},f)=\frac{ \vect{\pi} \exp \left( \mat{T} \displaystyle \int_0^x \mu_f(s)ds  \right)\vect{t} }{ \vect{\pi} \exp \left( \mat{T} \displaystyle\int_0^x \mu_f(s)ds  \right)\vect{e}}.
\end{align*}
It is noteworthy to remark that  although this method will yield visually better estimation of the mortality (or log-mortality) curves, the resulting PI model will have a lower likelihood than the one arising from the EM algorithm.

\subsection{Maximum likelihood estimation}

We now present estimation via maximum-likelihood (ML) for the PI model. We consider the estimation procedure for the general setting where the parameters of the inhomogeneity transform are also regressed, incurring in no additional mathematical complexity.

In many applications, a large proportion of the data is not entirely observed, or censored. In this paper, we consider the case of right-censoring, since it arises naturally in the context of its applications in survival analysis \footnote{{ 
We would like to remark that for the mortality applications considered in the sequel, the EM algorithm for right-censored data will not be used in its full strength, since our data will come from 1x1 mortality tables. However, Section \ref{sec:cens} illustrates the implementation of the right-censored algorithm for survival data.
}}
The main difference is that we no longer observe exact data points $Y = y$, but only $Y \in (v, \infty )$, conditionally on the corresponding covariate information. In particular, by monotonicity of $g$ we have that $g^{-1}(Y ; \vect{\theta} )\in ( g^{-1}(v ; \vect{\theta} ), \infty )$, which can be interpreted as a censored observation of a random variable with conventional PH distribution. Some of the formulas from \cite{olsson1996estimation} are useful to adapt an expectation-maximization (EM) algorithm for IPH distributions in the presence of covariates and censored data.

Let $(z_1,\delta_1), \dots, (z_N, \delta_N)$ be a possibly right-censored i.i.d.\ sample from of a PH distributed random variable $Z$ with parameters $(\bfpi , \bfT )$, where the $\delta_i$, $i=1,\dots,N$ are binary indicators taking the value $1$ if case of observation and $0$ in case of right-censoring. 
Now, for each $k,l\in\{1,\dots,p\}$, let $B_k$ be the number of times that the underlying jump-process $(J_t)_{t\geq0}$ initiates in state $k$, $N_{kl}$ the total number of jumps from state $k$ to $l$, $N_k$ the number of times that we reach the absorbing state $p+1$ from state $k$ and let $V_k$ be the total time that the underlying Markov jump process spends in state $k$ prior to absorption. These statistics are hidden, in the sense that $Z$ alone cannot account for them.
If they were not hidden, and given a sample of absorption times $\vect{z}$, the completely observed likelihood $\mathcal{L}_c$ can be written in terms of these sufficient statistics as follows:
\begin{equation}
\mathcal{L}_c( \bfpi , \bfT ;\vect{z})=
\prod_{k=1}^{p} {\pi_k}^{B_k} \prod_{k=1}^{p}\prod_{l\neq k} {t_{kl}}^{N_{kl}}e^{-t_{kl}V_k}\prod_{k=1}^{p}{t_k}^{N_k}e^{-t_{k}V_k},   
\end{equation}
which can be brought into the canonical form of the exponential dispersion family of distributions, which possess explicit maximum likelihood estimators.

However, since the full data is not observed, we employ the EM algorithm to obtain the ML estimators. At each iteration the E-step consists in computing the conditional expectations of the sufficient statistics $B_k$, $N_{kl}$, $N_k$ and $V_k$ given that $Z=z$, for the fully observed case, and given that $Z>z$ for the right-censored case. The M-step maximizes  $\mathcal{L}_c( \bfpi , \bfT ,\vect{z})$ using the estimates of the sufficient statistics from the previous step, obtaining updated parameters $( \bfpi , \bfT )$.

Below we write the formulas needed for the E- and M-steps for a sample of size $N$. We denote by ${\bfe}_k$ the column vector with all elements equal to zero besides the $k$-th entry which is equal to one, that is, the $k$-th element of the canonical basis of $\mathbb{R}^d$.

We now consider a sample of size $N$. For any two sets, $A,B$ define their product $A\times B=\{(a,b)|\,a\in A,\,b\in B\}$, with the definition extending inductively to the case of more than two sets. Define the following product set $\mathcal{Z}=\mathcal{Z}(\vect{z})=\prod_{i=1}^N A_i$, where $A_i=\{z_i\}$ if $\delta_i=1$, and $A_i=(z_i,\infty)$ otherwise. Then it can be shown that the conditional expectations can be written as follows:

\begin{enumerate} 
\item[ 1)]\textit{E-step, conditional expectations:} 
\begin{align*}
    \mathbb{E}(B_k\mid \mat{Z}\in\mathcal{Z})=\sum_{i=1}^{N}\left\{ \delta_i\frac{\pi_k {\bfe_k}^{ \mathsf{T}}\exp( \bfT z_i) \bft }{ \bfpi \exp( \bfT z_i) \bft }+(1-\delta_i)\frac{ \pi_{k} \bfe_{k}^{\mathsf{T}} \exp\left(\mat{T} z_i\right)\bfe } { \bfpi \exp\left(\mat{T} z_i\right)\bfe}\right\},
\end{align*}
\begin{align*}
 \mathbb{E}(V_k\mid\mat{Z}\in\mathcal{Z})&=\sum_{i=1}^{N} \left\{\delta_i\frac{\int_{0}^{z_i}{\bfe_k}^{ \mathsf{T}}\exp( \bfT (z_i-u)) \bft  \bfpi \exp( \bfT u)\vect{e}_kdu}{ \bfpi \exp( \bfT z_i) \bft } \right.\\
  &\quad\quad \left. + (1-\delta_i)\frac{\int_{0}^{z_i}{\bfe_k}^{ \mathsf{T}}\exp( \bfT (z_i-u)) \bfe  \bfpi \exp( \bfT u)\vect{e}_kdu}{ \bfpi \exp( \bfT z_i) \bfe }\right\},
\end{align*}
\end{enumerate}
\begin{align*}
\mathbb{E}(N_{kl}\mid\mat{Z}\in\mathcal{Z})&=\sum_{i=1}^{N}t_{kl} \left\{\delta_i \frac{\int_{0}^{z_i}{\bfe_l}^{\mathsf{T}}\exp( \bfT (z_i-u)) \bft \, \bfpi \exp( \bfT u)\vect{e}_kdu}{ \bfpi \exp( \bfT z_i) \bft } \right.\\
&\quad\quad \left.+ (1-\delta_i) \frac{\int_{0}^{z_i}{\bfe_l}^{\mathsf{T}}\exp( \bfT (z_i-u)) \bfe \, \bfpi \exp( \bfT u)\vect{e}_kdu}{ \bfpi \exp( \bfT z_i) \bfe }\right\},
\end{align*}
\begin{align*}
\mathbb{E}(N_k\mid \mat{Z}\in\mathcal{Z})=\sum_{i=1}^{N} \delta_it_k\frac{ \bfpi  \exp( \bfT z_i){\bfe}_k}{ \bfpi \exp( \bfT z_i) \bft }.
\end{align*}
\begin{enumerate}

\item[2)] \textit{M-step, explicit maximum likelihood estimators:} 
\begin{align*}
\hat \pi_k=\frac{\mathbb{E}(B_k\mid \mat{Z}\in\mathcal{Z})}{N }, \quad \hat t_{kl}=\frac{\mathbb{E}(N_{kl}\mid \mat{Z}\in\mathcal{Z})}{\mathbb{E}(V_{k}\mid \mat{Z}\in\mathcal{Z})}
\end{align*}
\begin{align*}
\hat t_{k}=\frac{\mathbb{E}(N_{k}\mid \mat{Z}\in\mathcal{Z})}{\mathbb{E}(V_{k}\mid \mat{Z}\in\mathcal{Z})},\quad \hat t_{kk}=-\sum_{s\neq k} \hat t_{ks}-\hat t_k.
\end{align*}
We set $$ \hat{\bfpi}=(\hat{\pi}_1, \dots,\: \hat{\pi}_p)^{},\quad \hat{\bfT} =\{{ \hat{t}}_{kl}\}_{k,l=1,2,\dots,p},\quad\hat{\bft}=( \hat{t}_1,\dots,\: \hat{t}_p)^{\mathsf{T}}.$$
\end{enumerate}
Note that the contributions from the conditional expectation of $N_{k}$ given that $Z>z_i$ are zero, since absorption has not yet taken place. 

For convenience, a summary of the entire procedure is provided in Algorithm \ref{alg;IPHreg}.

\begin{algorithm}[]
\caption{Full EM algorithm for the Proportional Intensities (PI) model}\label{alg;IPHreg}
\begin{algorithmic}
\State \textit{\textbf{Input}: positive data points $\vect{y}=(y_1,y_2,\dots,y_N)^{\mathsf{T}}$, censoring indicators $(\delta_1,\delta_2,\dots,\delta_N)^{\mathsf{T}}$, covariates $\vect{x}_1,\dots,\vect{x}_N$, and initial parameters $( \bfpi , \bfT , \vect{\beta},\vect{\gamma})$}\\
\begin{enumerate} 
\item[ 1)]\textit{Transformation:} Transform the data into $z_i=g^{-1}(y_i \,|\, \vect{x}_i, \vect{\gamma})$, $i=1,\dots,N$.

\item[ 2)]\textit{E-step:} Define the set $\mathcal{Z}=\mathcal{Z}(\vect{z})=\prod_{i=1}^N A_i$, where $A_i=\{z_i\}$ if $\delta_i=1$, and $A_i=(z_i,\infty)$ otherwise. The, compute the statistics

\begin{align*}
    \mathbb{E}(B_k\mid \mat{Z}\in\mathcal{Z})=\sum_{i=1}^{N}\left\{ \delta_i\frac{\pi_k {\bfe_k}^{ \mathsf{T}}\exp( \bfT z_i) \bft }{ \bfpi \exp( \bfT z_i) \bft }+(1-\delta_i)\frac{ \pi_{k} \bfe_{k}^{\mathsf{T}} \exp\left(\mat{T} z_i\right)\bfe } { \bfpi \exp\left(\mat{T} z_i\right)\bfe}\right\},
\end{align*}
\begin{align*}
 \mathbb{E}(V_k\mid\mat{Z}\in\mathcal{Z})&=\sum_{i=1}^{N} \left\{\delta_i\frac{\int_{0}^{z_i}{\bfe_k}^{ \mathsf{T}}\exp( \bfT (z_i-u)) \bft  \bfpi \exp( \bfT u)\vect{e}_kdu}{ \bfpi \exp( \bfT z_i) \bft } \right.\\
  &\quad\quad \left. + (1-\delta_i)\frac{\int_{0}^{z_i}{\bfe_k}^{ \mathsf{T}}\exp( \bfT (z_i-u)) \bfe  \bfpi \exp( \bfT u)\vect{e}_kdu}{ \bfpi \exp( \bfT z_i) \bfe }\right\},
\end{align*}
\begin{align*}
\mathbb{E}(N_{kl}\mid\mat{Z}\in\mathcal{Z})&=\sum_{i=1}^{N}t_{kl} \left\{\delta_i \frac{\int_{0}^{z_i}{\bfe_l}^{\mathsf{T}}\exp( \bfT (z_i-u)) \bft \, \bfpi \exp( \bfT u)\vect{e}_kdu}{ \bfpi \exp( \bfT z_i) \bft } \right.\\
&\quad\quad \left.+ (1-\delta_i) \frac{\int_{0}^{z_i}{\bfe_l}^{\mathsf{T}}\exp( \bfT (z_i-u)) \bfe \, \bfpi \exp( \bfT u)\vect{e}_kdu}{ \bfpi \exp( \bfT z_i) \bfe }\right\},
\end{align*}
\begin{align*}
\mathbb{E}(N_k\mid \mat{Z}\in\mathcal{Z})=\sum_{i=1}^{N} \delta_it_k\frac{ \bfpi  \exp( \bfT z_i){\bfe}_k}{ \bfpi \exp( \bfT z_i) \bft }.
\end{align*}

\item[3)] \textit{M-step: let} 
\begin{align*}
\hat \pi_k=\frac{\mathbb{E}(B_k\mid \mat{Z}\in\mathcal{Z})}{N }, \quad \hat t_{kl}=\frac{\mathbb{E}(N_{kl}\mid \mat{Z}\in\mathcal{Z})}{\mathbb{E}(V_{k}\mid \mat{Z}\in\mathcal{Z})}
\end{align*}
\begin{align*}
\hat t_{k}=\frac{\mathbb{E}(N_{k}\mid \mat{Z}\in\mathcal{Z})}{\mathbb{E}(V_{k}\mid \mat{Z}\in\mathcal{Z})},\quad \hat t_{kk}=-\sum_{l\neq k} \hat t_{kl}-\hat t_k.
\end{align*}

\item[4)] Maximize
	\begin{align*}
	(\hat{\vect{\beta}},\hat{ \vect{\gamma}})  & = \argmax_{( \vect{\beta}, \vect{\gamma})}\left\{ \sum_{i=1}^N  \delta_i\log (\overline{F}_{Y}(y_i; \hat{\bfpi}, \hat{\bfT} ,\vect{\beta}, \vect{\gamma})) +(1- \delta_i)\log (f_{Y}(y_i; \hat{\bfpi}, \hat{\bfT} ,\vect{\beta}, \vect{\gamma})) 
 \right\}\\
 & = \argmax_{(\vect{\beta},\vect{\gamma})} \left\{\sum_{i=1}^N  \delta_i\log\left(
 m(\vect{x}_i \vect{\beta} ) \lambda (y;  \vect{\theta}(\vect{x}_i \vect{\gamma} ) )\, \hat{\bfpi}\exp \left( m(\vect{x}_i \vect{\beta} )\int_0^y \lambda (s;  \vect{\theta}(\vect{x}_i \vect{\gamma} ) )ds\ \hat{\mat{T}} \right)\hat{\vect{t}} 
 \right) \right.\\
 &\quad\quad\quad\quad\quad\quad\left. 
 +(1- \delta_i)\log\left(
  \hat{\bfpi}\exp \left( m(\vect{x}_i \vect{\beta} )\int_0^y \lambda (s;  \vect{\theta}(\vect{x}_i \vect{\gamma} ) )ds\ \hat{\mat{T}} \right)\vect{e} 
 \right)
 \right\}.
\end{align*}

\item[5)] Update the current parameters to $({\bfpi},\mat{T} ,\vect{\beta},\vect{\gamma}) =(\hat{{\bfpi}},\hat{\mat{T}},\hat{\vect{\beta}},\hat{\vect{\gamma}})$. Return to step 1 unless a stopping rule is satisfied.
\end{enumerate}
    \State \textit{\textbf{Output}: fitted parameters of the proportional intensities model $( \bfpi , \bfT , \vect{\beta},\vect{\gamma})$.}
\end{algorithmic}
\end{algorithm}

%
%
%
%

\begin{remark}\rm 
	Algorithm~\ref{alg;IPHreg} generalizes the EM algorithm for IPH distributions in \cite{albrecher2020iphfit} in the following two settings: i) presence of censored observations; and ii) incorporation of covariate information. 
	
\end{remark}
 
	\begin{remark}\rm
	The E-steps require the computation of matrix exponentials which can be performed in multiple ways \citep{moler1978nineteen}. In \cite{asmussen1996fitting}, this is done by converting the problem into a system of ODEs, which are then solved via a Runge-Kutta method of fourth-order (a C implementation, called EMpht, is available online \cite{olsson1998empht}). We employed a new implementation in Rcpp of such a method for the illustration here and everywhere else in the paper. An important consideration is that the Runge-Kutta method requires the sample to be ordered. This is relevant  since
the transformed data of the above algorithm, $z_i=g^{-1}(y_i \,|\, \vect{x}_i, \vect{\gamma})$ will in general not be ordered anymore, given that covariate information varies across subjects.  Thus, intermediate ordering steps have to be introduced before proceeding to the numerical implementation of the EM algorithm of \cite{olsson1996estimation} by the Runge-Kutta method. The re-ordering also modifies the ties in the data, which must be accounted for. Additionally, the evaluation of the likelihood in Step 2 may be done using the by-products of Step 1, which implies that re-ordering and re-weighting will also be necessary at each iteration within the optimization inside Step 2.
These additional steps pose a significant computational burden relative to the non-covariate algorithm.
The uniformization method \citep{albrecher2020iphfit} is an alternative, but the comparison of the two methods' performance is beyond the scope of this paper.
\end{remark}

The following assertion is a general property of any EM algorithm.

\begin{proposition}
Employing Algorithm \ref{alg;IPHreg}, the likelihood function increases at each iteration. Since the likelihood is bounded for fixed $p$, it is guaranteed to converge to a (possibly local) maximum.
\end{proposition}
\section{Examples}\label{sec:examples}
All examples in this section were preliminarily fitted with the \texttt{c++} wrapper functions \texttt{fit} (for the no-covariates case) and \texttt{reg} (for the PI model), available in the \texttt{matrixdist} package in \texttt{R}, cf. \cite{bladt2021matrixdist,matrixdist}. Subsequently, the direct mortality modeling approach with $L(\mu,\nu)=(\log(\mu)-\log(\nu))s^2$ was implemented in the same language, and optimized with the \texttt{fortran} wrapper function \texttt{optim}.

The preliminary fitting was done using as log-likelihood weights the implied density from the mortality rates (death to exposure ratio), and using the midpoints between ages as the observed ages (corresponding to the data $\vect{y}$). An interval-censored approach would be slightly more accurate, but is outside the scope of the current paper, since we are using the EM procedure only to identify good initial parameters. For numerical purposes, age was divided by $100$ at the time of fitting, and all given parameters are on that scale. All plots are scaled back to the original size. All data was obtained from \url{https://www.mortality.org/}. {\color{black} In terms of running times, all analyses were completed individually within two minutes to 1.5 hours\footnote{{\color{black} All computations were done in a standard MacBook Pro (15-inch, 2017) Laptop with a 2.8 GHz Quad-Core Intel Core i7 processor and a 16 GB 2133 MHz LPDDR3 memory.}}, with the more complex examples being lengthier to optimize.}

\subsection{Danish female mortality without covariates}
We first consider the modeling of Danish females since the year $2000$ using a generalized Coxian structure as well as a Coxian structure. The distinction is made since the former seems to provide better fitted models, while the latter has a possibly more intuitive interpretation in terms of state transitions (as all individuals start in the first state).

The number of phases was chosen empirically. The left panel of Figure \ref{Danish_females} shows with a dashed line the $p=1$ case, corresponding to a classical Gompertz distribution. Thereafter, increasing dimensions of $\bfT$ were checked, for $p\le 20$. At some point, additional parameters do not provide increased performance\footnote{{\color{black}Increased performance in this context means a substantial decrease in loss function value upon convergence for each model of varying dimensions. Information criteria such as AIC and BIC are not helpful here, since the effective degrees of freedom of PH models are generally unknown (see also \cite{Alaric2022}). Thus we chose a relative decrease of less than $0.1\%$ as the stopping rule.}}, and we choose the last $p$ before this occurs. We observe a satisfactory fit, and the parameters for the generalized Coxian case are given by\footnote{Here, and in subsequently reported parameters, zeros are only exactly zero when the structure implies it (such as a Coxian structure). Else, they represent the rounding of a small number.}
\begin{align*}
\bfpi&=(0.01, 0, 0.04, 0.07, 0, 0.23, 0.06, 0.25, 0.32, 0, 0, 0, 0, 0.01, 0),\\
\mbox{Diag}(\bfT)&=-(3.14,1.81,95.92,25.29,24.27,2.98,0.03,0,17.19,0.13,1.23\times 10^{12},7.28,3.29,0.01),\\
\bft&=(0, 0.42, 1.77, 0, 0.01,  0,  0, 0, 0, 10.16, 0.13, 492.2, 0, 0, 0.01)^\prime,\\
\beta&=10.68 .
\end{align*}

\begin{figure}[!htbp]
\centering
\includegraphics[width=0.49\textwidth]{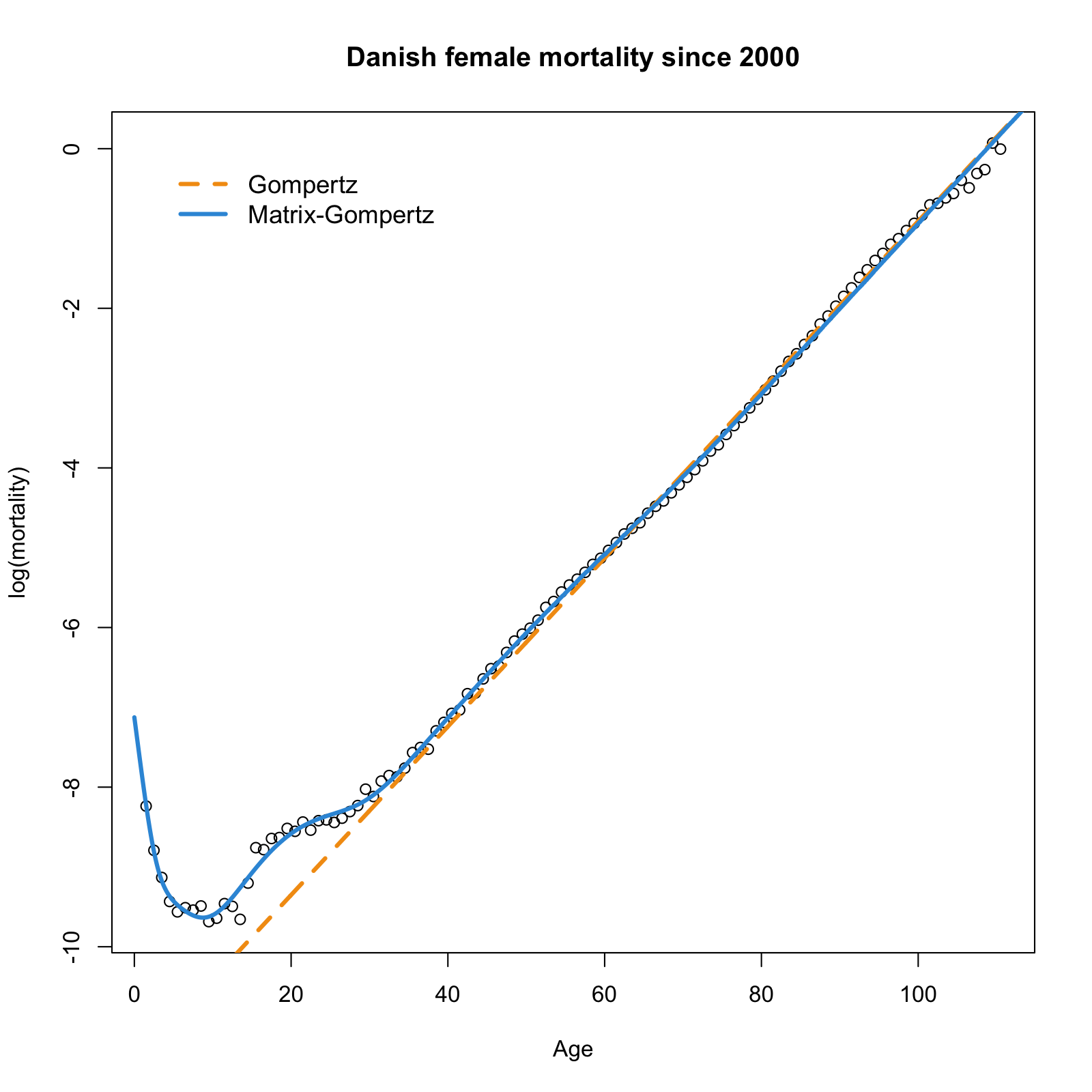}
\includegraphics[width=0.49\textwidth]{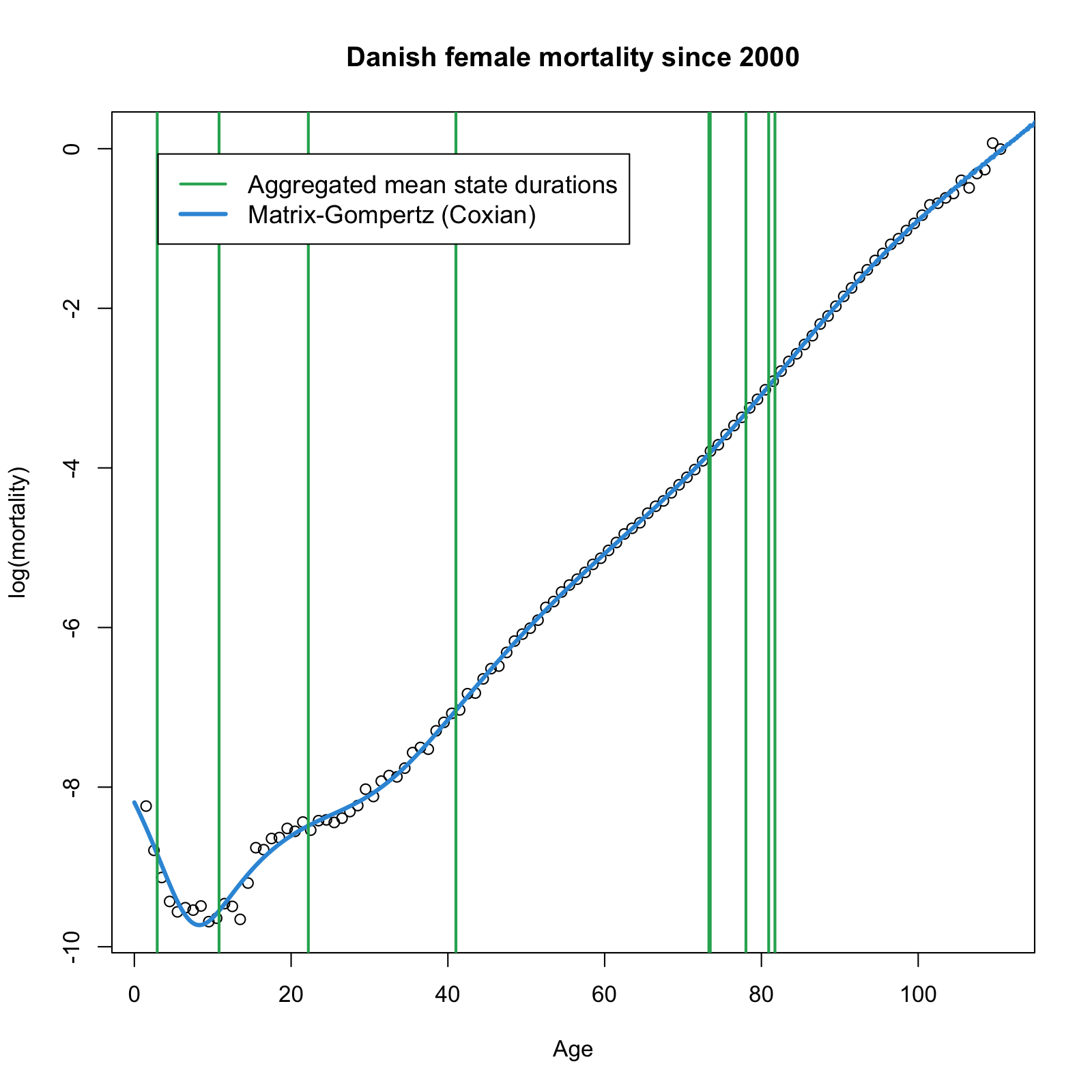}
\caption{Fitted matrix-Gompertz distributions to Danish female mortality data (generalized Coxian PH on the left, Coxian PH on the right).
} \label{Danish_females}
\end{figure}
Note that for the classical Gompertz fit on the left-hand side one obtains $\beta=10.55$, but one is not able to take into account the humps at the left end at all. On the other hand, the additional flexibility of the generalized Coxian can accomplish that, maintaining a  similar value $\beta=10.68$ for the transformation.
	
While the choice of the generalized Coxian is motivated by the tractability of the corresponding estimation procedure (it is much more tractable than fitting a general PH distribution) within the purpose of obtaining a good statistical fit, the restriction to a classical Coxian distribution (starting in state 1 with probability 1) opens up the interpretability in terms of stages of aging as discussed in Section \ref{sec:int}. The right panel of Figure \ref{Danish_females} depicts the resulting Coxian fit to the same dataset, together with the cumulative aggregated means of the sojourn times for each state of the underlying state space (here 10 states). In the Coxian construction, these are traversed sequentially, and thus the last vertical line is equal to the life expectancy. The parameters and line locations are given by
\begin{align*}
\bfpi&=\vect{e}_1,\\
\mbox{Diag}(\bfT)&=-(27.72, 5.95, 1.84, 0.29, 0.02, 2.32, 2.88\times 10^{11},0.04, 0.03, 0.04),\\
\bft&=(0.03,  0,   0.01,   0,   0.01,   0, 189.94,   0.01,   0.01,   0.04)^\prime,\\
\mbox{lines}&=(2.92, 10.81, 22.20, 41.02, 73.30, 73.45, 73.45, 78.02, 80.93, 81.73),
\\
\beta&=7.85.
\end{align*}
One can see that the wavy behavior of the log-mortality curve is also captured quite well with a pure Coxian PH distribution. Recall that the Coxian construction allows for the interpretation that life traverses through different states until death (absorption), staying at each one for a random duration, and with many individuals going through all the states, whereas some individuals die (go to absorption) from an earlier state directly. From the above figures, one readily computes the probabilities of dying while being in state $i$ ($i=1,\ldots,10$):
$$(0.001, 0, 0.003, 0.006, 0.326, 0.001, 0, 0.217, 0.428,1).$$ The probability of reaching state $p=10$ before death (i.e., going through all states) here is $0.299$.
Observe that some states have a very short duration, which can serve as an artifact to generate possible deaths at these concrete points in time. Note also that the Coxian matrix-Gompertz leads to a nice fit, but 'needs' a different value $\beta=7.85$ for the transformation than the classical Gompertz distribution in order to fit the humps at the left end. The generalized Coxian matrix-Gompertz then is flexible enough again to return to a similar $\beta$-value as the classical Gompertz, which is determined by the right end. From the concrete resulting parameters, one sees that there is a high probability to start in state 5, 7 or 8 only, so that the resulting model can be interpreted in terms of more flexible mixtures of Coxian distributions with lower state space. 




\subsection{Female mortality with country as a covariate}
Let us now move on to apply our approach to multi-population fitting with regression. As a first illustration of the PI model, we consider the mortality curves from two countries with substantially different life expectancies, the United States of America and Japan. Thus, the covariate is a binary one: USA $=1$ and Japan $=0$. We study the female population in the time interval $2000-2010$ (the observations are aggregated). 
The same approach for selecting the order of the fitted model as above was employed for this dataset. Figure \ref{Japan_USA_females} shows that although there is a slight misfit at the teenage years, the overall shape of both fitted curves resemble closely the observed mortality. Here we have $p=12$ states and the parameters are given by
\begin{align*}
\bfpi&=(0.07, 0.03, 0, 0, 0, 0.15, 0, 0.05, 0.24, 0.45, 0, 0),\\
\mbox{Diag}(\bfT)&=-(0.02, 0.06, 22.37,   8.78,   0.20,   1.37,   1.22,  44.17, 404.86,   0, 31.44, 0),\\
\bft&=(0,  0, 13.89,  0.25,  0,  0,  0.02,  0,  0,  0,  0.35,  0)^\prime,\\
\vect{\beta}&=(0.91),\\
\vect{\theta}&=(2.28,-0.07).
\end{align*}

\begin{figure}[!htbp]
\centering
\includegraphics[width=0.7\textwidth]{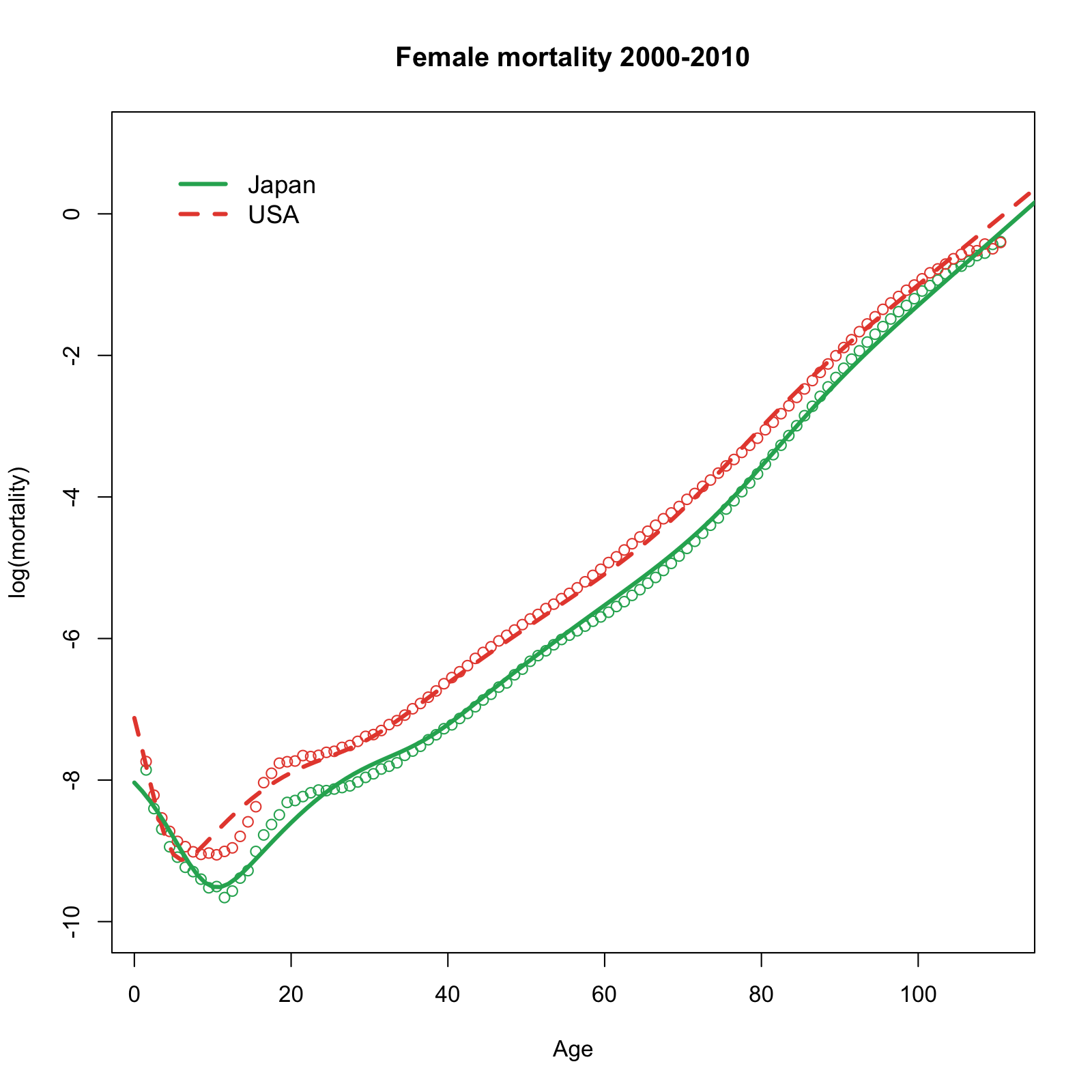}
\caption{ PI model applied to country as a covariate.
} \label{Japan_USA_females}
\end{figure}

The positive value of $\vect{\beta}$ indicates an increased underlying intensity (and thus shorter sojourn times) for the USA population, which is multiplicative with factor $\exp(0.91)\approx 2.48$. On the other hand, the parameters of $\vect{\theta}$ indicate that there is a baseline Gompertz factor of size $\exp(2.28)\approx 9.78$, and the USA population has a multiplicative decrease of this factor by $\exp(-0.07)$, or roughly $93\%$, incurring in an overall Gompertz factor of size $9.16$. The latter means in terms of the PI model that there is an ``environment" lifetime generating Gompertz process which is unfavorable for the Japanese female population, but their underlying slower traversing through the multi-state process (of their lives) makes up for it, to the degree that can be observed in Figure \ref{Japan_USA_females}.

Theoretically, by Theorem \ref{aysmpt_proportional}, both curves are asymptotically parallel, and although we see a substantial difference at most ages, the gap closes and is quite small at very advanced ages. This observation can be made from the data itself, implying that old-age mortality in both countries seems to be quite similar. However, the takeaway from the hidden Markov point of view is that the flexibility gained from considering intensities instead of hazard rates is advantageous for modeling more delicate inter-curve features in regression analyses.

\subsection{Danish female mortality with time as a covariate}
As a second example of the PI model, let us consider the Danish female population during the interval $1950-2000$, and examine the effect of time when considering it as a numerical covariate. For numerical convenience, we actually use the transformed covariate $t^\ast=(\mbox{year}-1950)/1000$. The same approach as in the two previous examples is used to select the size $p$ of the underlying state space. 

The proportional hazards model is not well suited for this task (the resulting log-mortality curves for different time periods would be parallel!), and so we use the Lee-Carter (LC, cf. \cite{lee1992modeling}) model as a benchmark. The LC model specifies the log-mortality by 
\begin{align*}
\log(\mu_{x,t})=a_x+b_x k_t+\epsilon_{x,t},
\end{align*}
where $\epsilon_{x,t}$ are Gaussian random variables. The $a_x$ term is estimated as the average log-mortality over time at each age $x$, and then $b_x$ and $k_t$ are computed from a singular value decomposition of $\log(\mu_{x,t})-a_x$. The procedure is non-parametric and thus can capture very minute changes in individual mortality curves across different ages.

On the other hand, by Theorem \ref{aysmpt_proportional} the PI model asymptotically satisfies, as $x$ grows,
\begin{align*}
\log(\mu_{x,t})\to a +b_{x,t}+ k_t,
\end{align*}
where $a=\log(c)$ is a constant which depends on the parameters $\bfpi$ and $\bfT$, $b_{x,t}=\log(\lambda(x;\exp(\theta_0+\theta_1 t^\ast)))=\exp(\theta_0+\theta_1 t^\ast)\log(x)$ (in our case) is the logarithm of the time-dependent Gompertz intensity, and $k_t=\beta_1 t^\ast$ (also in our case) is the logarithm of the proportionality factor between the underlying intensities. For smaller $x$ (before convergence to the limit) the log-mortality does not have such a simple decomposition, and is instead given by
\begin{align}
	\log(\mu_{x,t})=b_{x,t} &+ \log\left(\vect{\pi}\exp \left( \int_0^x  \exp(b_{s,t}+k_t)ds\ \mat{T} \right)\vect{t} \right)\label{mortality_time_ph}\\
	& -\log\left(\vect{\pi}\exp \left( \int_0^x \exp(b_{s,t}+k_t)ds\ \mat{T} \right)\vect{e} \right)  \,.\nonumber
	\end{align}

Figure \ref{DNK_females_time} shows the resulting fitted PI and LC models for the selected years $1960$, $1980$, and $2000$, where in both cases we observe a downward trend in mortality every $20$ years. Both models have room for improvement, but overall show accurate descriptions. The PI model seems to not capture a surge in mortality for old ages in $1960$, while the LC model does not seem to capture the decreases of mortality from $1980$ to $2000$ for young ages. The non-parametric nature of the LC model is sturdier to outliers as is the case for an old-age point of $1960$, while the PI seeks to compromise large parts of the curve for this one observation. However, observe that the LC model has some wiggling occurring for ages $0$ to $40$, which may indicate a slight degree of overfitting, a known issue for smaller populations and one which the PI model seems to handle better in this case. 

The parameters of the PI model are as follows:
\begin{align*}
\bfpi&=(0.31, 0.15, 0.17, 0.11, 0.12, 0.01, 0, 0.01, 0.04, 0.03, 0.04, 0),\\
\mbox{Diag}(\bfT)&=-(36.47, 0.09,   1.3,   0.06,   9.49,   5.34,   0.19,   2.25, 139.68,   0.22, 0.89,   0.66),\\
\bft&=(0, 0, 0, 0, 0, 0.17, 0, 0, 4.16, 0, 0, 0.66)^\prime,\\
\vect{\beta}&=(-22.54),\\
\vect{\theta}&=(1.92,2.94).
\end{align*}

\begin{figure}[!htbp]
\centering
\includegraphics[width=\textwidth]{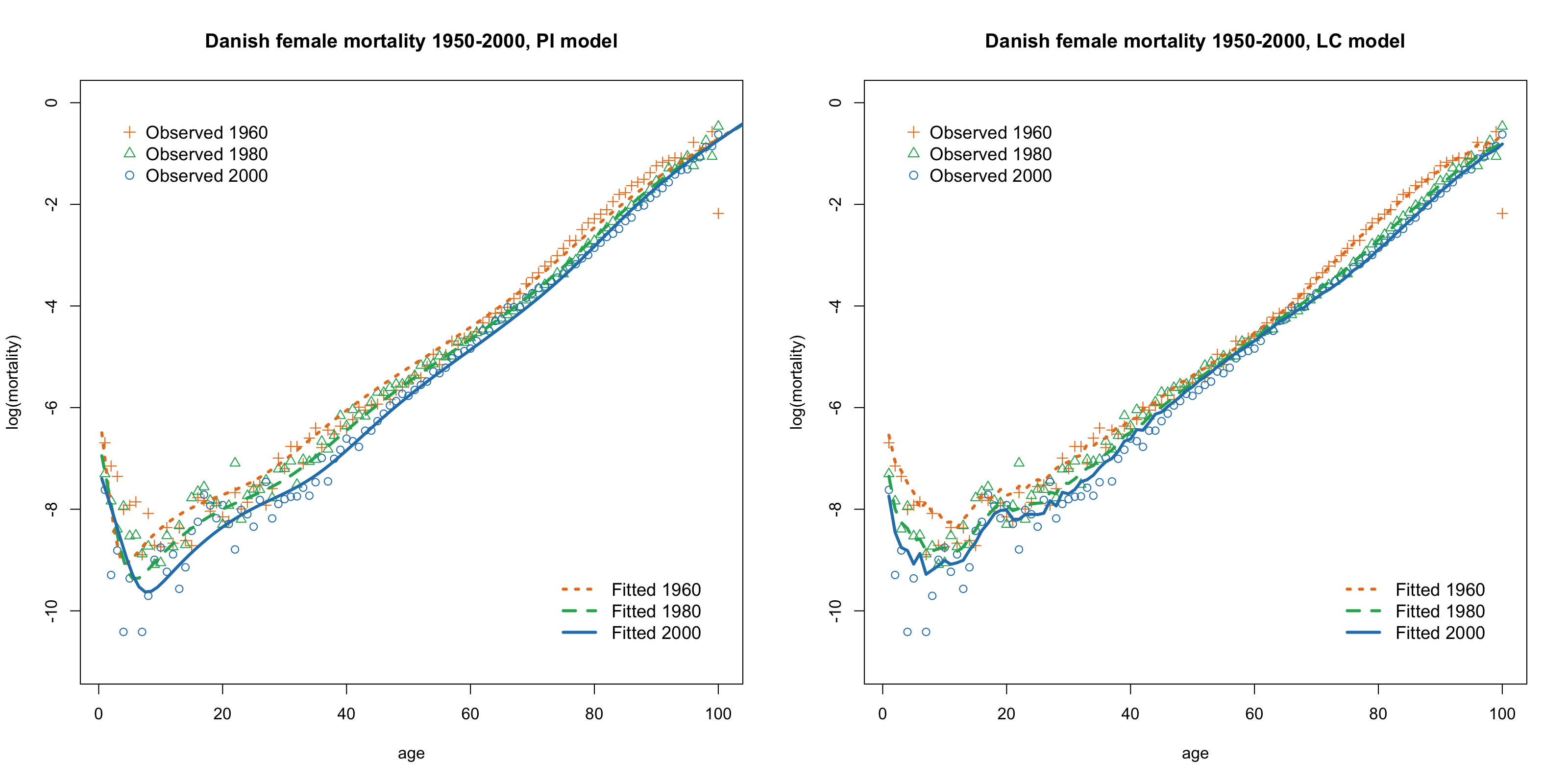}
\caption{ PI model using time as a covariate, plotted for {\color{black} the years 1960, 1980 and 2000, and for the Danish females dataset}.
} \label{DNK_females_time}
\end{figure}

In contrast to the previous illustration, we have that the negative value of $\vect{\beta}$ indicates a decreased underlying intensity (and thus longer sojourn times) as years progress, which is multiplicative with factor $\exp(-22.54\cdot t^\ast)$. In turn, the parameters of $\vect{\theta}$ imply a baseline Gompertz factor given by $\exp(1.92)\approx 6.82$ at $t^\ast=0$ (which amounts to the year $1950$), thereafter it increases with the factor $\exp(2.94 \cdot t^\ast)$, which itself increases the mortality as time progresses. Figure \ref{DNK_females_time} shows that the overall tradeoff from increased ``environmental" mortality and decreased state-traversing speed leads to a net decrease in mortality through the years at young ages, but virtually no changes at older ages.

In order to better understand the issue of potential overfitting of the LC model for smaller populations (like the case of Denmark above), we did an analogous analysis for USA females, and the resulting parameters of the PI model in this case are 
\begin{align*}
	\bfpi&=(0.3, 0.15, 0.13, 0.13, 0.15, 0.01, 0, 0.01, 0.05, 0.03, 0.04, 0),\\
	\mbox{Diag}(\bfT)&=-(1.15,  0.11,  0.65,  0.05,  7.82,  4.02,  0.19,  2.81, 56.14,  0.2,  0.77,  0.70),\\
	\bft&=(0, 0, 0, 0, 0, 0.22, 0, 0, 2.72, 0, 0, 0.7)^\prime,\\
	\vect{\beta}&=(-17.98),\\
	\vect{\theta}&=(1.93, 1.74). 
\end{align*}
Figure \ref{USA_females_time} depicts the resulting log- mortality curves for 1960, 1980 and 2000 together with the empirical observations. Indeed, here the wiggling of the LC curves is considerably improved. One sees that in this case 
from a purely statistical point of view the LC model may be preferable over the PI model, but one should keep in mind that the PI model offers a more direct interpretation of the resulting model, with much fewer parameters, while still providing a {\color{black} visually acceptable statistical fit}. 

\begin{remark}\rm
{\color{black} Forecasting with such a model is straightforward if we are only interested in non-autoregressive effects\footnote{Given a desired time $t$ (and thus $t^\ast$) and age $x$, Equation \eqref{mortality_time_ph} provides the corresponding mortality estimate. Moreover, all the implied distributional properties can also be extracted using the fitted parameters, using Equation \eqref{density_of_model}.}, and the incorporation of cohort effects can also be done as with time, and with possible non-linear terms. However, the predictive performance falls short of models which are specifically designed for this task, and incorporating these temporal dependencies is an interesting subject of future research.}
\end{remark}

\begin{figure}[!htbp]
\centering
\includegraphics[width=\textwidth]{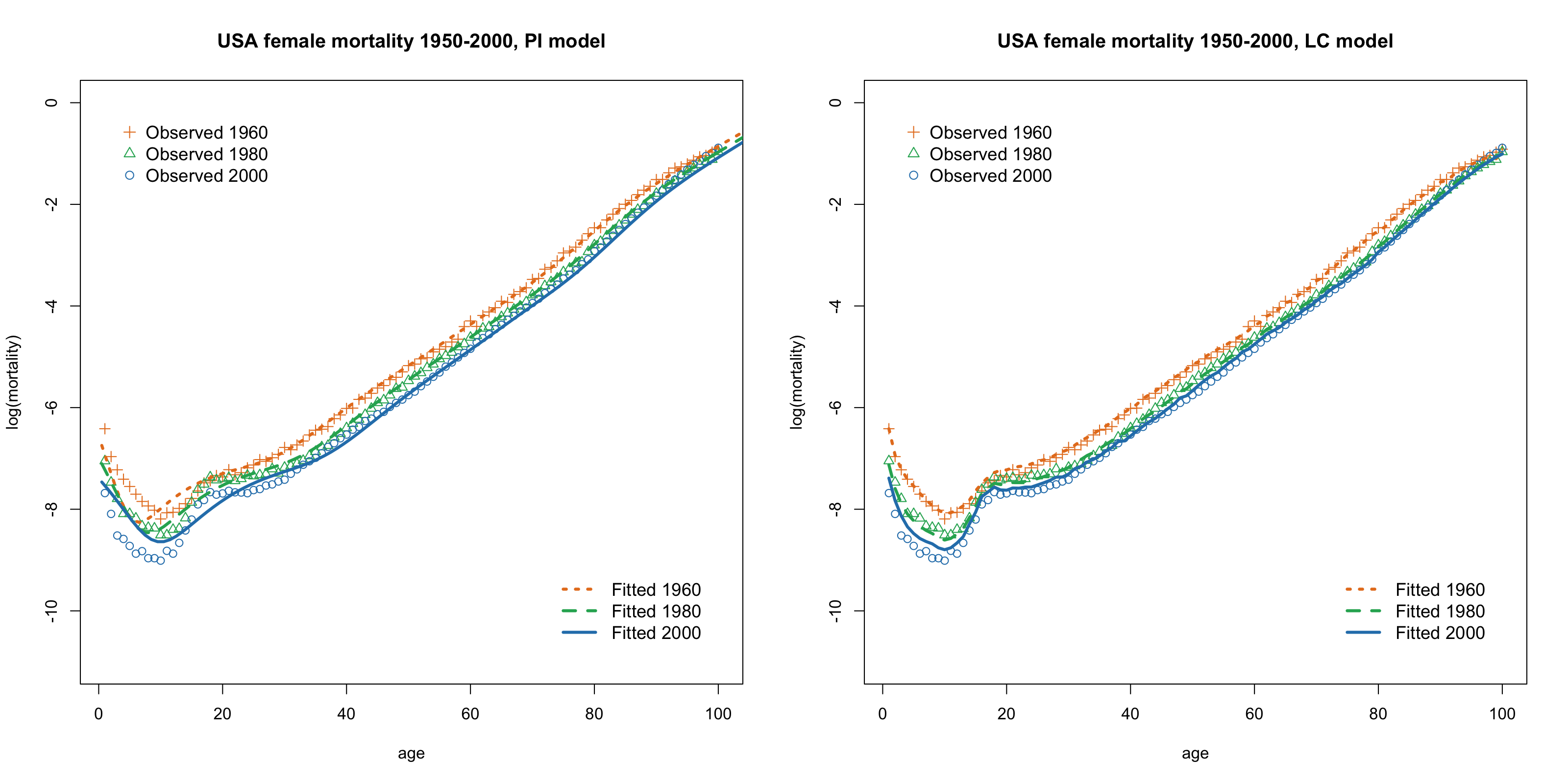}
\caption{ PI model using time as a covariate, plotted for {\color{black} the years 1960, 1980 and 2000, and for the USA females dataset}. 
} \label{USA_females_time}
\end{figure}

{ 
\section{General survival modeling}\label{sec:cens}
This section illustrates the use of inhomogeneous Markov models for survival analysis, showcasing the implementation of Algorithm \ref{alg;IPHreg} in the right-censored case. The first is a simulated example, while the second is a Veterans' lung cancer dataset, both of which are subject to right-censoring and include covariate information, and thus the regression methods from this paper can be applied.

\subsection{Simulated data}
Consider the following simple two-group setup. We study two groups of equal size, $500$ each. For the first group,
$$Z_i\sim \mbox{IPH}\left((1/4,\,1/2,\, 1/4), \, \begin{pmatrix}
-10 & 0 & 0\\
0 & -1 & 0\\
0 & 0 & -1/10
\end{pmatrix},\, \frac{3}{2}s^{1/2}\right)$$
while for the second
$$Z_i\sim \mbox{PH}\left((1/4,\,1/2,\, 1/4), \, \begin{pmatrix}
-20 & 0 & 0\\
0 & -2 & 0\\
0 & 0 & -2/10
\end{pmatrix}\right).$$
Subsequently, we randomly censor the $Z_i$ by $1000$ exponential variables $E_i$ of rate $1/10$ by 
$$Y_i=\min\{Z_i,\,E_i\}, \quad \delta_i=1\{Y_i=Z_i\},\quad i=1,\dots, 1000,$$
leading to roughy $10\%$ right-censored observations. Notice that in the above setting, the inhomogeneity function also changes with the covariates, such that we need to implement the fitting procedure associated with model \eqref{double_prop_intens_spec}. Additionally, for comparison, we fitted two misspecified models: a standard Weibull proportional hazards model and the scalar ($\mbox{dim}(\mat{T})=1$) version of model \eqref{double_prop_intens_spec}. The resulting empirical versus fitted survival curves for each group, as well as the hazard ratio between the two groups, are given in Figure \ref{sim_ex}. We observe that only the correctly specified model recovers the oscillating shape of the theoretical hazard ratio, and within each group, the survival curve is better estimated in the tails of the distribution. 
Additionally, the goodness of fit can be evaluated by considering a generalization of the Cox-Snell residuals given by 
\begin{align*}
r_i=-\log\left(\vect{\pi}\exp \left(m(\vect{x}_i \vect{\beta} ) \int_0^{y_i} \lambda (s;\, \vect{\theta}(\vect{x}_i \vect{\gamma}))ds\ \mat{T} \right)\vect{e}\right)\,,\quad i=1,\dots, 100 \,,
\end{align*}
which under the null hypothesis of correct specification of the model, should constitute a right-censored unit-exponential dataset. Figure \ref{sim_ex2} depicts the Kaplan-Meier and Nelson-Aalen non-parametric estimators of the survival curve of the residuals $r_i$ to assess standard exponential behavior visually. We observe that only the correctly specified model provides an adequate fit in the sense of staying within the confidence bounds implied by Greenwood's formula for the first approach, and by aligning itself with the identity for the second one.

\begin{figure}[!htbp]
\centering
\includegraphics[width=0.45\textwidth]{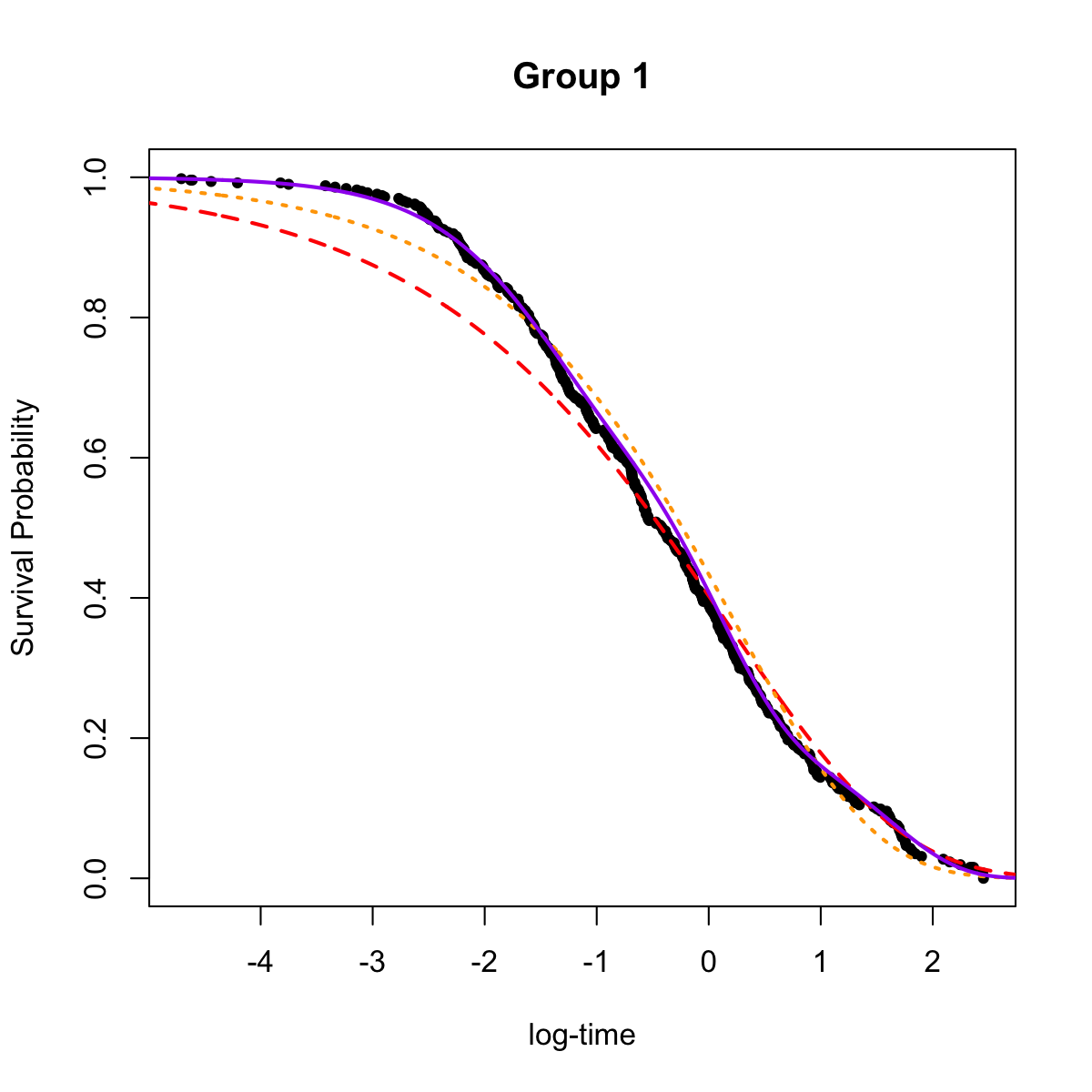}
\includegraphics[width=0.45\textwidth]{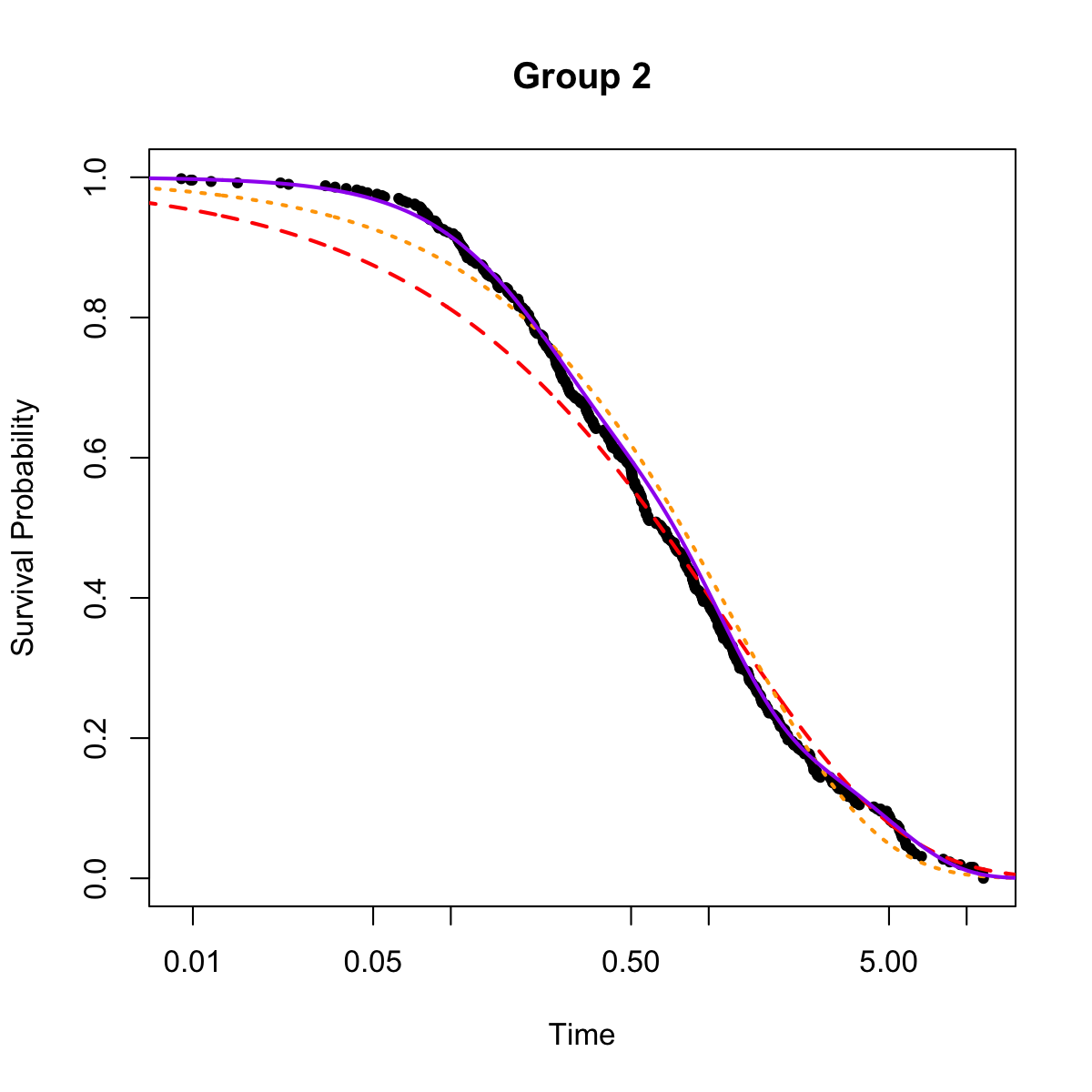}
\includegraphics[width=0.45\textwidth]{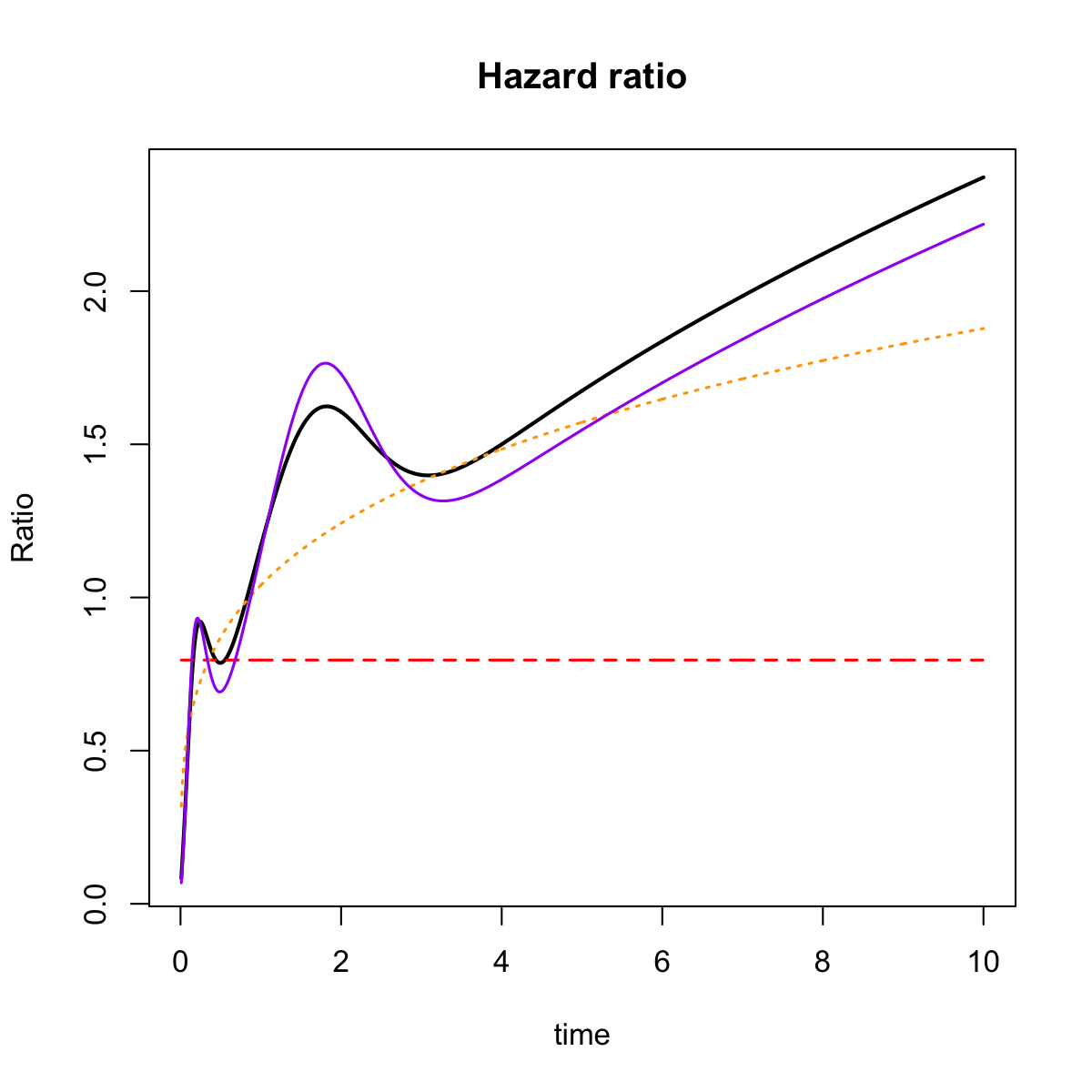}
\caption{Simulated example. Top: empirical (black, dots) versus fitted survival curves for Weibull proportional hazards model (red, dashed), scalar version of \eqref{double_prop_intens_spec} (orange, dotted), and correctly specified \eqref{double_prop_intens_spec} (purple, solid). Bottom: theoretical hazard function ratio between the two groups (black solid)  versus the aforementioned models (same line code).
} \label{sim_ex}
\end{figure}

\begin{figure}[!htbp]
\centering
\includegraphics[width=\textwidth]{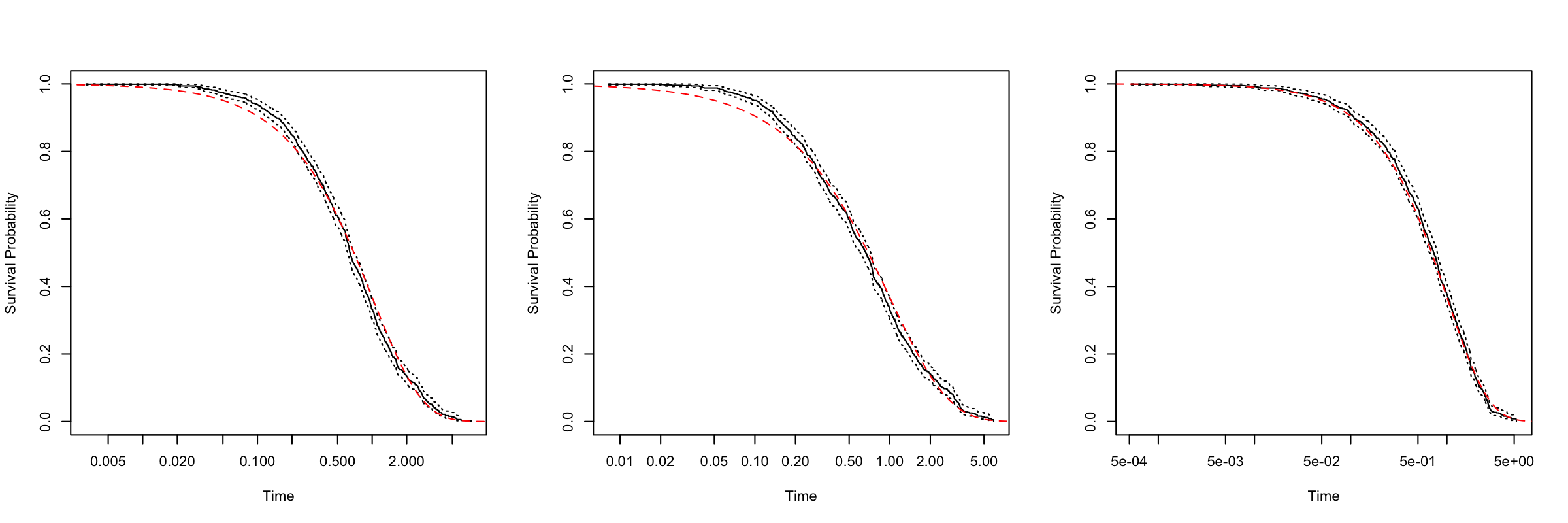}
\includegraphics[width=\textwidth]{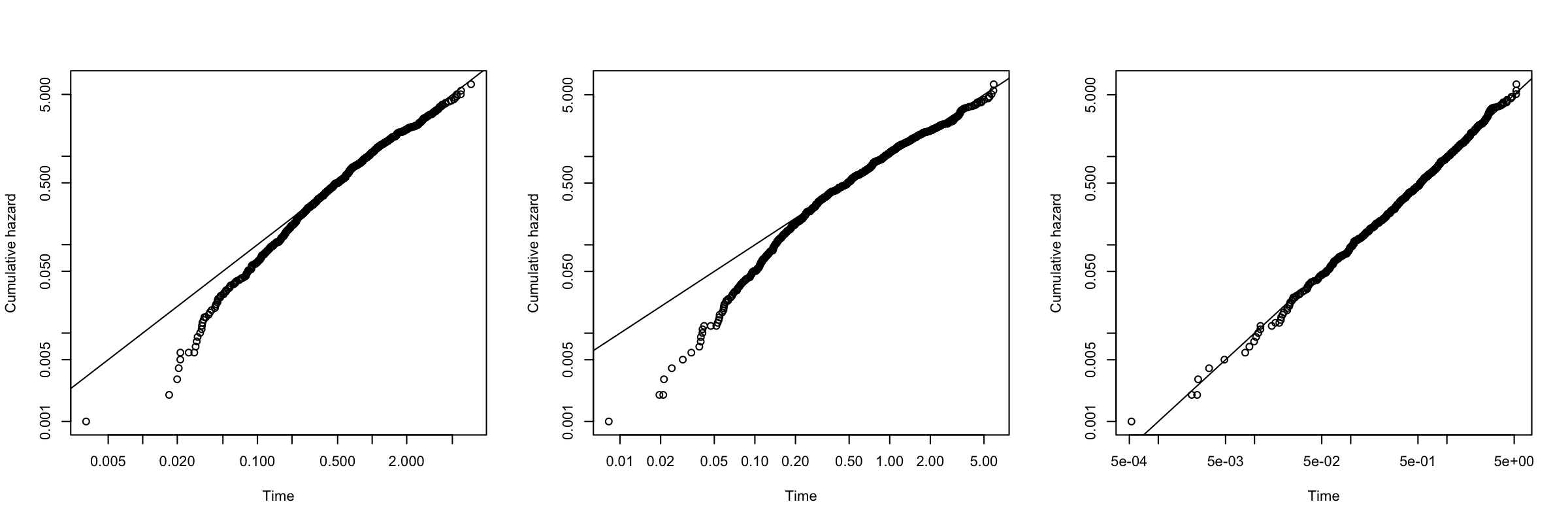}
\caption{Simulated example goodness of fit based on the three models from Figure \ref{sim_ex}. Top: survival curves based on the Kaplan-Meier estimator of the residuals $r_i$. Bottom: plot between $r_i$ and $-\log(\widehat{\overline{F}}(r_i))$ based on the Nelson-Aalen estimator.
} \label{sim_ex2}
\end{figure}

\subsection{Veterans data}
We consider a well-known dataset from the Veterans' Administration Lung Cancer study (publicly available in the R package \textit{survival}).

The data consist of $137$ observations and $12$ variables coming from a randomized trial where two groups are given different treatment regimens for lung cancer. The variables of interest for our analysis are survival time, censoring status (binary), treatment (binary), prior therapy (binary), and Karnofsky performance score (from 0 to 100). Previous studies (cf. for instance, the vignette \url{https://cran.rstudio.com/web/packages/survival/vignettes/timedep.pdf}, pp. 15-22) have shown that only the Karnofsky score violates the proportional hazards assumption. Common ways to address this issue are stratifying the dataset into time groups or considering continuous time-dependent regression coefficients. However, the coefficient for treatment becomes significant when improving the model's fit. Consequently, we presently focus only on the goodness of fit of the model introduced in the text. The latter is measured by its log-likelihood and residuals.

For comparison, we consider a classical model: a) Weibull proportional hazards model. Given the modest size of the dataset, we avoid large matrices $\mat{T}$ for the IPH-based model. Thus, we present a model based on $\dim({\mat{T}})=2$ with a Coxian structure: b) Matrix-Weibull PI model.

The results are summarized in the table below and in Figure \ref{vets}.

\begin{table}[H]
\begin{tabular}{|c|c|c|l|l|}
\hline
Model & Log-likelihood & Nb. Parameters & AIC        & BIC        \\ \hline
a   & -136.21   & 5          & 282.42     & 297.02     \\ \hline
b   & -127.74   & 7          & {269.48} & {289.92} \\ \hline
\end{tabular}
\end{table}

We observe that model b) improves the fit to where the residuals are within bounds, and the log-likelihood, AIC and BIC reflect this. 

The AIC and BIC computed here for the matrix-distributions works well for comparison purposes, given the low dimensions of the distributions. However, in general, one should proceed with caution (especially in high dimensions) due to the well-known identifiability issue for PH distributions, which carries over to IPH survival models. Namely,  different dimensions and parameter configurations may lead to very similar density shapes.

\begin{figure}[]
\centering
\includegraphics[width=0.7\textwidth]{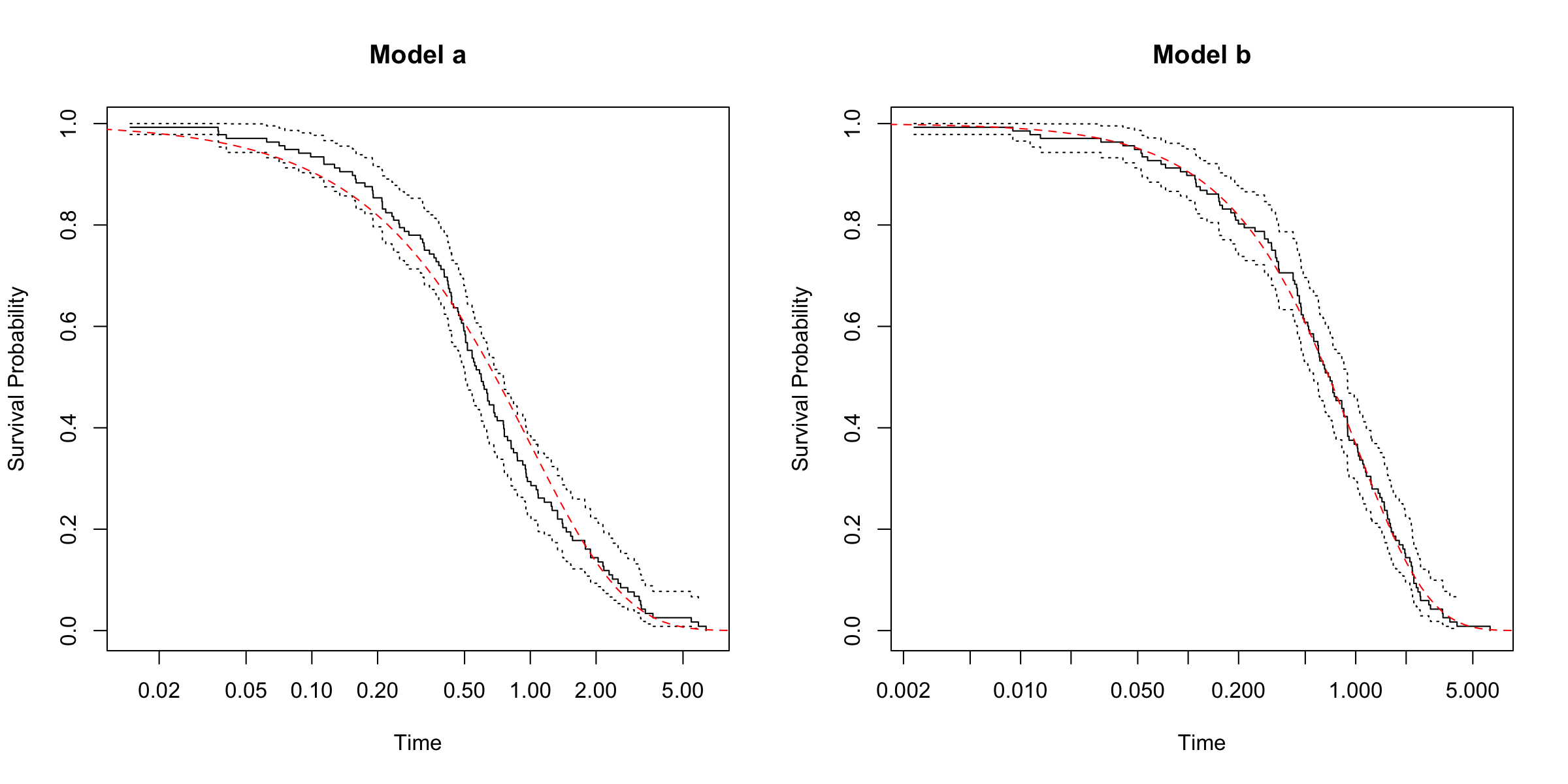}
\includegraphics[width=0.7\textwidth]{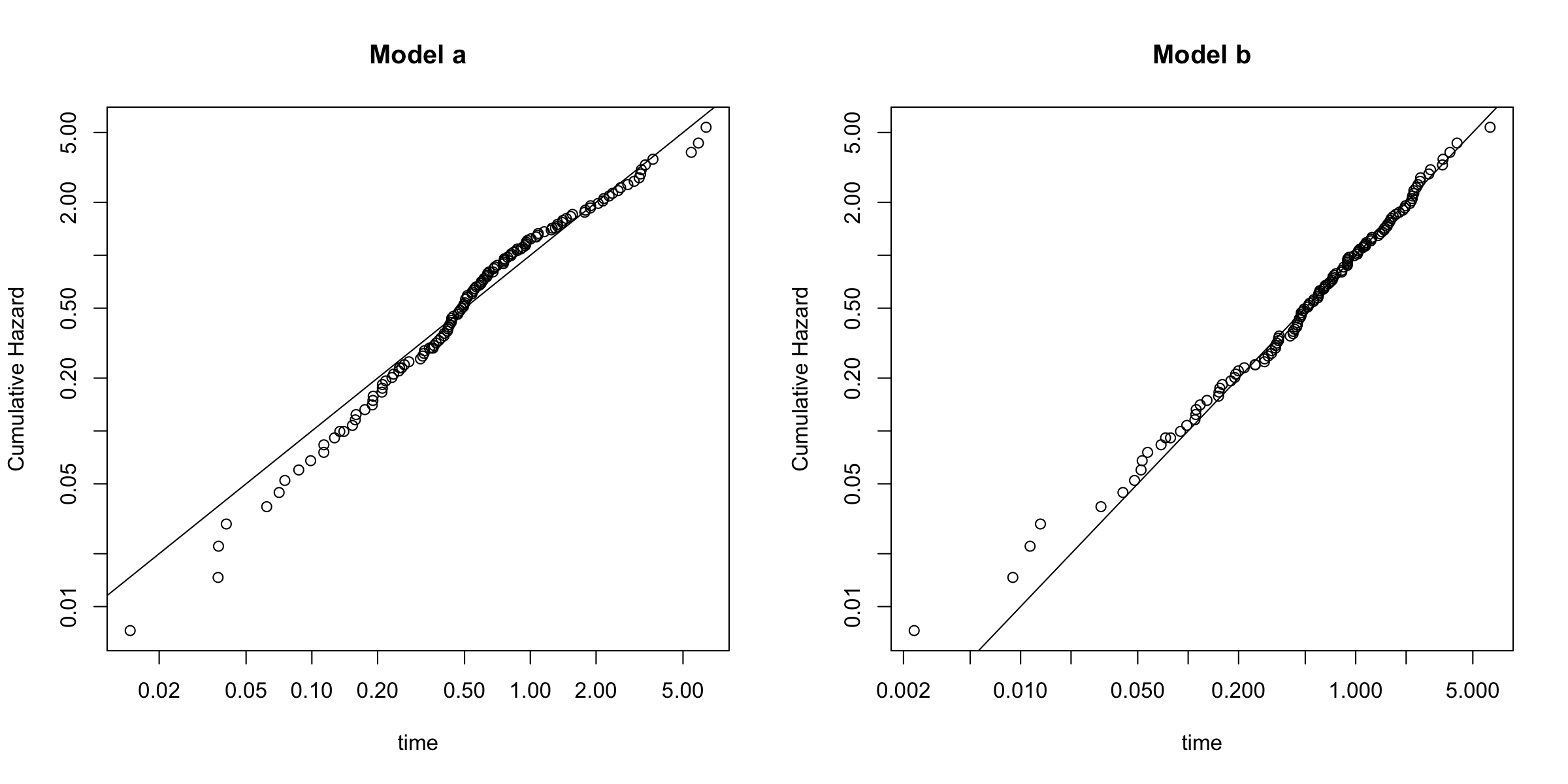}
\caption{Veterans dataset. Goodness of fit diagnostics based on the Kaplan-Meier estimator of the residuals (top row), and on the residuals versus Nelson-Aalen implied cumulative hazard plot (bottom row). The corresponding models are: a) Weibull proportional hazards model; b) Matrix-Weibull PI model.} \label{vets}
\end{figure}

}

\section{Conclusion}\label{sec:conclusion}
In this paper we investigated the potential of matrix-Gompertz distributions as special cases of inhomogeneous phase-type distributions for the modeling of mortality across the entire lifespan as a parsimonious alternative to popular modeling approaches. The goal is not to replace the latter but rather to see how close in performance one can get with this type of models, as they also allow a causal interpretation in terms of biological aging, with lifetime traversing through different stages and a possible early death in each of them. In that latter context, the introduction of the time-transform inherent in the construction of inhomogeneous phase-type distributions allowed to reduce the necessary dimension for adequate modeling reported in previous studies from 200 - 250 to 10 - 12. In addition, we studied a regression approach that can be used for multi-population mortality modeling based on proportional intensities. As the first contribution on regressing IPH distributions, this part may also be of independent interest. We developed an estimation procedure and illustrated the models for female mortality data from Denmark, Japan and the USA. 

Since matrix-Gompertz distributions are dense in the class of all lifetime distributions, from a general perspective we have illustrated that misspecified survival models can sometimes be just one matrix away from providing a good fit. The practical implementation of the methods in this paper relies on the EM algorithm, which for large matrices and datasets can become significantly slow. It is our experience that the use of an inhomogeneous intensity function in the underlying stochastic jump process can alleviate the need for a large number of phases, thus making estimation feasible.

Although we have used the midpoint approach, which is standard and works well for $1\times 1$ resolution life tables, an interval-censored estimation procedure is statistically more accurate and has potential to be applied to data at lower resolutions, which will be the subject of future research. In the present paper we considered regression according to proportional intensities, as a first and workable approach for IPH regression. {  We also showcased how the right-censored version of our algorithm can be used in survival settings.} In general, for multi-population mortality modeling it is of course restrictive to assume that the operational time for each population is scaled with a constant across the entire lifetime. Another promising direction is hence the regression of individual rates of the intensity matrix, as this would allow for interesting conclusions in terms of the aging interpretation underlying the IPH models. However, if applied to all possible parameters, the model will be highly overparametrized, and hence some constraints will have to be imposed, challenging the corresponding estimation methods. It will be an interesting subject of future research to identify workable compromises in this direction, which should increase the versatility of the resulting models.\\

%

\textbf{Acknowledgment.} The authors would like to thank Johannes Thuswaldner for conscientiously implementing some tests for the fitting procedures at an early stage of the project, and Peter Hieber for interesting discussions on the topic. HA and MaB would like to acknowledge financial support from the Swiss National Science Foundation Project 200021\_191984. JY would like to acknowledge financial support from the Swiss National Science Foundation Project  IZHRZ0\_180549.
\bibliography{MortalityModelling.bib}

\end{document}